\documentclass[journal]{IEEEtran}
\usepackage{cite}
\usepackage{amsmath,amssymb,amsfonts}

\usepackage{amsthm}
\usepackage{epsfig}
\usepackage{color}
\usepackage{lettrine}
\usepackage{lipsum}
\usepackage{bm}
\usepackage{subcaption}
\usepackage{mathtools}
\usepackage{graphics}
\usepackage{bbm}
\usepackage{etoolbox}
\usepackage{suffix}
\usepackage{booktabs}
\usepackage[T1]{fontenc}
\usepackage{textcomp}
\usepackage{xcolor}
\usepackage{stackengine}

\def \treq {\stackrel{\tiny \Delta}{=}}

\newtheorem{lemma}{Lemma}
\newcommand{\q}{\ensuremath{\mathbf q}}
\newcommand{\Q}{\ensuremath{\mathbf Q}}
\newcommand{\Qi}{\ensuremath{{Q}}^{-1}}
\newcommand{\N}{\ensuremath{\mathcal N}}
\newcommand{\e}{\mathbf{e}}
\newcommand{\E}{\ensuremath{\mathbb E}}
\newcommand{\R}{\ensuremath{\mathbb R}}
\newcommand{\p}{\ensuremath{\mathbf p}}
\renewcommand{\P}{\ensuremath{\mathbf P}}

\newcommand{\LL}{\ensuremath{\mathbf L}}
\newcommand{\HH}{\ensuremath{\mathbf H}}
\newcommand{\OO}{\ensuremath{\mathcal  O}}

\newcommand{\Prob}[1]{(\textbf{P}{#1})}
\newcommand{\cst}[1]{\textbf{C}{#1}}
\newcommand{\s}{\ensuremath{\mathbf s}}

\newcommand{\prop}{\text{JTRD}}
\newcommand{\rd}{\text{RDFT}}
\newcommand{\td}{\text{TDFR}}
\newcommand{\east}{\ensuremath{\mathrm{EAST}}}

\def \treq {\stackrel{\tiny \Delta}{=}}

\makeatletter
\newcommand{\removelatexerror}{\let\@latex@error\@gobble}
\makeatother

\usepackage{scalerel}
\usepackage{tikz}
\usetikzlibrary{svg.path}
\definecolor{orcidlogocol}{HTML}{A6CE39}
\tikzset{
  orcidlogo/.pic={
    \fill[orcidlogocol] svg{M256,128c0,70.7-57.3,128-128,128C57.3,256,0,198.7,0,128C0,57.3,57.3,0,128,0C198.7,0,256,57.3,256,128z};
    \fill[white] svg{M86.3,186.2H70.9V79.1h15.4v48.4V186.2z}
                 svg{M108.9,79.1h41.6c39.6,0,57,28.3,57,53.6c0,27.5-21.5,53.6-56.8,53.6h-41.8V79.1z M124.3,172.4h24.5c34.9,0,42.9-26.5,42.9-39.7c0-21.5-13.7-39.7-43.7-39.7h-23.7V172.4z}
                 svg{M88.7,56.8c0,5.5-4.5,10.1-10.1,10.1c-5.6,0-10.1-4.6-10.1-10.1c0-5.6,4.5-10.1,10.1-10.1C84.2,46.7,88.7,51.3,88.7,56.8z};
  }
}

\newcommand\orcidicon[1]{\href{https://orcid.org/#1}{\mbox{\scalerel*{
\begin{tikzpicture}[yscale=-1,transform shape]
\pic{orcidlogo};
\end{tikzpicture}
}{|}}}}

\definecolor{darkpastelgreen}{rgb}{0.01, 0.75, 0.24}
\usepackage{float}
\usepackage[colorlinks=true, citecolor = darkpastelgreen]{hyperref} 
 \usepackage[nameinlink,noabbrev]{cleveref}
\usepackage{algorithmic}
\usepackage[ruled,norelsize]{algorithm2e}

\makeatletter
\def\ps@IEEEtitlepagestyle{%
  \def\@oddfoot{\mycopyrightnotice}%
  \def\@oddhead{\hbox{}\@IEEEheaderstyle\leftmark\hfil\thepage}\relax
  \def\@evenhead{\@IEEEheaderstyle\thepage\hfil\leftmark\hbox{}}\relax
  \def\@evenfoot{}%
}
\def\mycopyrightnotice{%
  \begin{minipage}{\textwidth}
  \centering \scriptsize
  \copyright~2023 IEEE. Personal use of this material is permitted.  Permission from IEEE must be obtained for all other uses, in any current or future media, including reprinting/republishing this material for advertising or promotional purposes, creating new collective works, for resale or redistribution to servers or lists, or reuse of any copyrighted component of this work in other works.
  \end{minipage}
}
\makeatother
\begin{document}

\bstctlcite{IEEEexample:BSTcontrol}

\title{Secure Short-Packet Communications via UAV-Enabled Mobile Relaying: Joint Resource Optimization and 3D Trajectory Design}


\author{Milad Tatar Mamaghani\textsuperscript{\orcidicon{0000-0002-3953-7230}},~\IEEEmembership{Member,~IEEE}, Xiangyun Zhou\textsuperscript{\orcidicon{0000-0001-8973-9079}},~\IEEEmembership{Fellow,~IEEE}\\ Nan Yang\textsuperscript{\orcidicon{0000-0002-9373-5289}},~\IEEEmembership{Senior Member,~IEEE}, and A. Lee Swindlehurst\textsuperscript{\orcidicon{0000-0002-0521-3107}},~\IEEEmembership{Fellow,~IEEE}
\vspace*{-3mm}\thanks{This work was supported by the Australian Research Council’s Discovery Projects funding scheme (Grant ID: DP220101318).}
\thanks{M. Tatar Mamaghani, X. Zhou, and N. Yang are with the School of Engineering, Australian National University, Canberra, ACT 2601, Australia (email: \href{mailto:milad.tatarmamaghani@anu.edu.au}{\textcolor{black}{milad.tatarmamaghani@anu.edu.au}}; \href{mailto:xiangyun.zhou@anu.edu.au}{\textcolor{black}{xiangyun.zhou@anu.edu.au}}; \href{mailto:nan.yang@anu.edu.au}{\textcolor{black}{nan.yang@anu.edu.au}}).}
\thanks{A. L. Swindlehurst is with the Center for Pervasive Communications and
Computing, Henry Samueli School of Engineering, University of California, Irvine, CA 92697, USA (email: \href{mailto:swindle@uci.edu}{\textcolor{black}{swindle@uci.edu}}).}
\thanks{A portion of this work has been presented at the IEEE Global Communications Conference, 4-8 December 2023, Kuala Lumpur, Malaysia \cite{mamaghani2023globecom}.}}

\markboth{}%
{Milad \MakeLowercase{\textit{et al.}}: Secure Short-Packet Communications via UAV-Enabled Relaying}

\maketitle

\begin{abstract}

Short-packet communication (SPC) and unmanned aerial vehicles (UAVs) are anticipated to play crucial roles in the development of 5G-and-beyond wireless networks and the Internet of Things (IoT). In this paper, we propose a secure SPC system, where a UAV serves as a mobile decode-and-forward (DF) relay, periodically receiving and relaying small data packets from a remote IoT device to its receiver in two hops with strict latency requirements, in the presence of an eavesdropper. This system requires careful optimization of important design parameters, such as the coding blocklengths of both hops, transmit powers, and the UAV’s trajectory. While the overall optimization problem is nonconvex, we tackle it by applying a block successive convex approximation (BSCA) approach to divide the original problem into three subproblems and solve them separately. Then, an overall iterative algorithm is proposed to obtain the final design with guaranteed convergence. Our proposed low-complexity algorithm incorporates robust trajectory design and resource management to optimize the effective average secrecy throughput of the communication system over the course of the UAV-relay’s mission. Simulation results demonstrate significant performance improvements compared to various benchmark schemes and provide useful design insights on the coding blocklengths and transmit powers along the trajectory of the UAV.

\end{abstract}


\IEEEpeerreviewmaketitle

\section{Introduction}
\IEEEPARstart{S}{hort-packet} communication (SPC) has recently emerged as a critical paradigm for meeting the stringent demands of massive machine-type communications (mMTC) and ultra-reliable and low-latency communications (uRLLC) for the support of the Internet of Things (IoT) in 5G-and-beyond wireless networks \cite{Durisi2016, Al-Fuqaha2015, Popovski2014}. Indeed, unlike conventional wireless communications, SPC generally involves the transmission of small data packets whose size can potentially go down to tens of bytes, wherein the length  of control information is comparable with that of the data payload.  This approach is gaining popularity in the context of IoT, since there is a growing demand for real-time and efficient communication between devices. SPC also plays an important role  in mission-critical applications such as intelligent transportation, telemedicine, and industrial automation, where stringent latency requirements are in place \cite{Schulz2017}.


One of the critical challenges of SPC is guaranteeing communication reliability, considering the fact that utilizing short packets for transmissions is inherently accompanied by severely degraded channel coding gain \cite{Yang2019a, Shirvanimoghaddam2019}. In addition,  SPC in IoT networks encounters severe security challenges, primarily due to the broadcast nature of wireless media as well as the mission-critical and private SPC data that IoT networks often need to share \cite{Feng2021}. SPC-IoT communication systems are vulnerable to eavesdropping. For example, in the case of intelligent transportation, if a message is intercepted by an adversary, private information such as user identity or location may be exposed. Thus, ensuring both the stringent reliability and security of SPC-IoT systems presents significant challenges. 
Also, resource-constrained IoT nodes have practical constraints such as limited energy, leading to the use of lightweight secrecy protocols 
to reduce the demand for resources. To this end, various physical-layer security (PLS) techniques have been investigated to provide protection through wiretap coding and smart signaling by exploiting the physical-layer properties of wireless channels \cite{Wang2019c, Sun2019b, Mukherjee2014, Hamamreh2018, Yang2022}.

Nonetheless, achieving the required secrecy by applying PLS techniques is nontrivial for mission-critical SPC-IoT applications. First, conventional PLS schemes and the adopted wiretap codes are primarily based on the assumption of infinite blocklength codes whose packet lengths are on the order of several kilobytes, as opposed to tens of bytes for SPC. Secondly, PLS-based designs for traditional large-packet communications are generally based on the so-called \textit{Secrecy Capacity} \cite{Poor2017}, from Shannon's classical information theory. However, secrecy capacity is not applicable for SPC scenarios due to the finite blocklength assumption, and therefore it is not appropriate to design classical PLS approaches for SPC scenarios since it could lead to suboptimal solutions. This calls for rethinking the analysis and design of PLS for the SPC-IoT systems.

Concurrently, unmanned aerial vehicles (UAVs) have recently become an increasingly popular technology 
for a variety of wireless and IoT applications. Particularly, UAVs are envisioned to be integrated into IoT systems to provide a range of benefits such as aerial base station (BS) or mobile relaying,  remote sensing and processing, real-time monitoring, mobile edge computing (MEC), etc.,  due to their relatively rapid on-demand deployment, inexpensive operation, low-cost maintenance, and flexible three-dimensional (3D) mobility \cite{Wu2021, Sun2021, Geraci2021, Akyildiz2020, li2022resource}. One of the main advantages of UAVs in IoT systems is their ability to gather data from difficult-to-reach geographic areas or hazardous locations. Also, UAVs are particularly useful in applications where timely data transmission is critical, such as emergency response and disaster relief efforts. Since they operate above ground away from obstacles, UAVs are effective at reducing signal attenuation in wireless links arising from shadowing and blockage effects. Accordingly, battery-limited IoT devices, enjoying favorable air-ground (AG) channel conditions, need significantly lower power to transmit data to UAVs, leading to a significant increase in their lifespan. However, security challenges such as eavesdropping for UAV-aided IoT communications are of significant importance and need to be meticulously addressed \cite{Jiang2021a, Wang2019, Wu2019}.

\subsection{Prior Studies and Motivation}
To safeguard different wireless networks, PLS techniques have been extensively studied in the literature for various applications such as cooperative non-orthogonal multiple access (NOMA) networks \cite{Yuan2019}, opportunistic relaying in IoT networks \cite{Ding2019}, full-duplex relaying \cite{Chen2015b}, energy harvesting-based device-to-device communications \cite{Xu2021b}, and UAV-aided secure wireless communication \cite{Mamaghani, guo2022ris, wang2022robust}. \textcolor{black}{In particular, \cite{guo2022ris} investigated resource allocation for secure downlink UAV transmissions towards a terrestrial user empowered by a reconfigurable intelligent surface (RIS) where a friendly UAV-jammer is employed for communication secrecy. The work in \cite{wang2022robust} explored improving PLS for a UAV-aided cognitive relay system with the help of cooperative jamming and robust resource management.
However, previous studies such as \cite{guo2022ris} and \cite{wang2022robust} have typically considered the design and optimization of infinite blocklength transmissions, which as mentioned earlier are not suitable for SPC-IoT scenarios due to their severely reduced channel coding gain and the impossibility of guaranteeing error-free transmissions.} Recently, there has been significant research interest in exploring to what extent finite blocklength transmissions incur a loss in capacity \cite{Polyanskiy2010}. Yang \textit{et al.} in their seminal work \cite{Yang2019} addressed the achievable secrecy rate (SR) for a general wiretap channel given fixed reliability and secrecy constraints in the finite blocklength regime, and analytically derived tight achievability and converse bounds on the maximum secret communication rate. Following these foundational results, the authors in \cite{Zheng2020} studied the design of PLS in SPC over fading channels by investigating secrecy throughput for both single and multiple antenna transmitters. In \cite{Feng2022}, the authors investigated a setting similar to \cite{Zheng2020} from another perspective, where they defined a new outage probability metric relevant to the characteristics of SPC. In \cite{Wang2019e}, the authors studied secure SPC in a mission-critical IoT system in the presence of an external multi-antenna eavesdropper. Further,  \cite{Li2022} explored secure transmission 
rates for downlink SPC with a queuing delay requirement under different assumptions on the channel state information (CSI), where the authors derived closed-form optimal power control policies for some special scenarios. We stress that the previous studies have considered system design with only static communication nodes and adopted a fixed number of information bits per short-packet transmission.  Thus, the frameworks developed in \cite{Zheng2020, Feng2022, Wang2019e, Li2022} may not be applicable for highly dynamic UAV-IoT scenarios or when a variable amount of data is generated from IoT devices for transmission.
In light of this, some recent research efforts have considered UAV-IoT network designs with downlink SPC, e.g., \cite{Wang2020b, Liu2023}. Nevertheless, blocklength optimization is completely ignored in the aforementioned works due to the assumption of simple one-hop transmissions. As a result, further research is necessary to fully understand the performance of SPC in UAV-IoT networks.

\subsection{Our Contributions}
\textcolor{black}{Inspired by the aforementioned research, in this work we study a secure UAV-aided relaying system in which sensitive short packets need to be periodically transmitted from a remote IoT device to a receiver with a stringent latency requirement while combating a passive eavesdropper in an uncertain location.  We design the network to provide robust communication resource allocation and 3D UAV trajectory over the duration of the mission to improve the overall system secrecy performance. Our contributions are detailed below.}
\textcolor{black}{ { \begin{itemize}
\item We propose a secure SPC-IoT system that exploits a mobile UAV as an aerial auxiliary relay with flexible 3D maneuverability to facilitate end-to-end confidential and reliable short-packet data exchange between distant IoT transceivers, effectively mitigate the risk of a ground eavesdropper in an uncertain location, and improve the reliability and secrecy performance of the network impacted by short-packet transmissions.
\item We formulate a new optimization problem in terms of the effective average
    secrecy throughput (EAST) for the considered UAV-aided SPC system under security, reliability, latency, and mobility constraints. The formulated problem is nonconvex with a nondeterministic objective function, and hence challenging to solve optimally.
    \item  To tackle the challenging nonconvex problem, we derive an analytical and tractable expression for a lower bound on the objective function. Then, based on the idea of block successive convex approximation (BSCA), we divide the original problem into three subproblems comprising joint power allocation, coding blocklength optimization, and 3D trajectory design, and we efficiently tackle each via convex optimization. We then propose a low-complexity algorithm combining these solutions to optimize the overall EAST performance.  We also discuss the complexity and convergence of our proposed algorithm.
    \item We conduct extensive simulations to offer useful insights into the performance of the proposed joint resource allocation and trajectory design, and highlight its benefits compared to other known benchmarks. We observe that the joint trajectory and resource optimization can significantly improve the EAST performance compared to designs with either fixed trajectory or predetermined communication power and coding blocklengths. 
    \item We investigate the impact of key system parameters on the overall system performance. In particular, we find that in our joint design, both the transmit power and uplink coding blocklength adaptively increase when the mobile UAV-assisted relay moves from the transmitter towards the receiver, while the relaying power maintains its maximum value with a decreasing downlink coding blocklength for EAST enhancement. Additionally, our joint design demonstrates high robustness in EAST regardless of the level of uncertainty in the eavesdropper's location.
\end{itemize}}}

The remainder of this paper is organized as follows. Section~\ref{sec:sysmodel} introduces our proposed UAV-aided relaying system model with imperfect location information about the eavesdropper, followed by the formulation of an optimization problem to improve the overall system performance in Section \ref{sec:problem}. In Section~\ref{sec:solution}, we present an efficient approach to tackle the optimization problem.  Section \ref{sec:numerical} discusses selected numerical results and the  impact of key system parameters. Finally, conclusions are drawn in Section~\ref{sec:conclusion}.

\section{System Model and Problem Statement}\label{sec:sysmodel}

\begin{figure}[!t]
\centering
\includegraphics[width=\columnwidth]{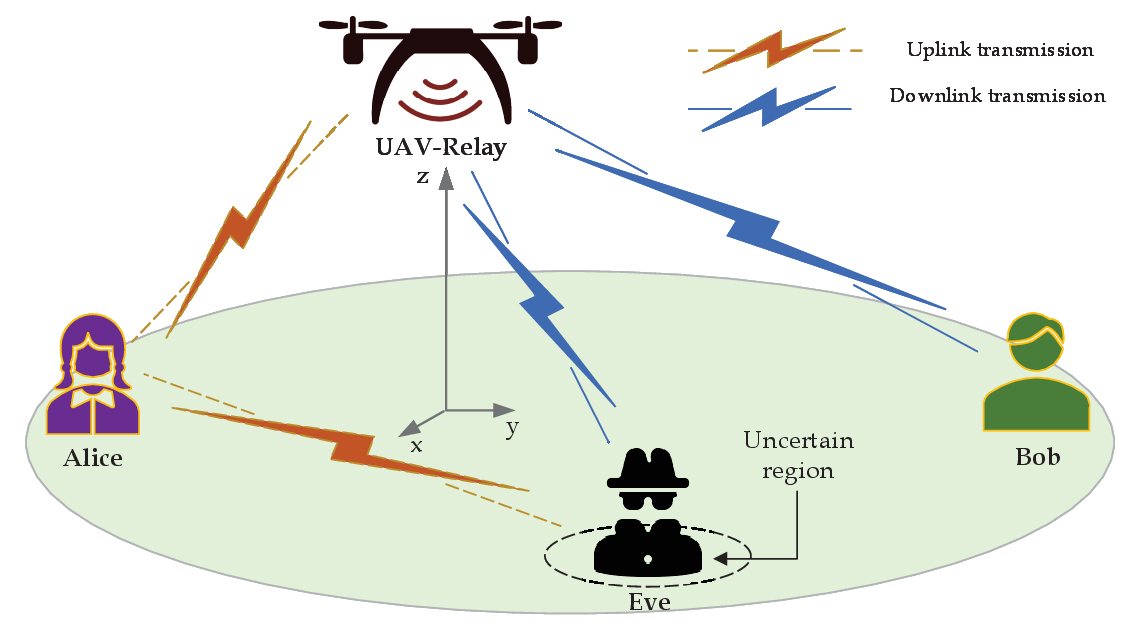}
\caption{System model for secure UAV-aided short-packet relaying in the presence of an adversary with location uncertainty.}
\label{fig1}
\end{figure}
We consider a UAV-assisted IoT communication system with secure SPC as illustrated in Fig. \ref{fig1}, wherein a transmitter (Alice) periodically sends short packets containing confidential information to a designated remote receiver (Bob) with the help of a trusted mobile UAV-Relay (UR) while a malicious passive eavesdropper (Eve), whose location  is not perfectly known, attempts to overhear the ongoing confidential transmissions. Here, SPC is considered due to its importance for various delay-sensitive IoT applications. In practice, Alice can be considered as an IoT device that periodically generates a short packet of sensitive information from the environment and, if feasible, immediately transmits it to a desired remote receiver Bob with a stringent latency requirement for monitoring, controlling or sensing applications. We refer to a \textit{timeslot} as a period of $\delta_t$ seconds, and assume that a packet is generated at the beginning of each timeslot. In this work, we allow Alice to generate and transmit a variable amount of information bits in each timeslot to accommodate different tasks, or to provide a suitable amount of information to Bob according to the communication channel conditions and other requirements.

\subsection{Periodic Secure Short-packet Relaying}
We assume that there exists no direct link between Alice and Bob due to the distance or blockages, and thus, a UAV-mounted relay, thanks to its flexible mobility and line-of-sight (LoS)-dominant AG channels, is employed to facilitate the end-to-end SPC. We assume that  all the communication nodes are equipped with only a single antenna, as commonly done for resource-constrained IoT devices.
In addition, the mobile relaying strategy adopted by the UR is assumed to be the decode-and-forward (DF) protocol with time division duplexing (TDD) on a shared bandwidth $W$ for both reception and transmission. Compared with amplify-and-forward (AF) relaying \cite{TatarMamaghani2018}, DF relaying is beneficial for the considered UAV-aided scenario since Bob is not required to obtain CSI for the first hop, which in turn results in lower overhead, and DF relaying avoids noise amplification, leading to generally better signal reception quality at Bob.

\begin{figure}[!t]
\centering
\includegraphics[width=\columnwidth]{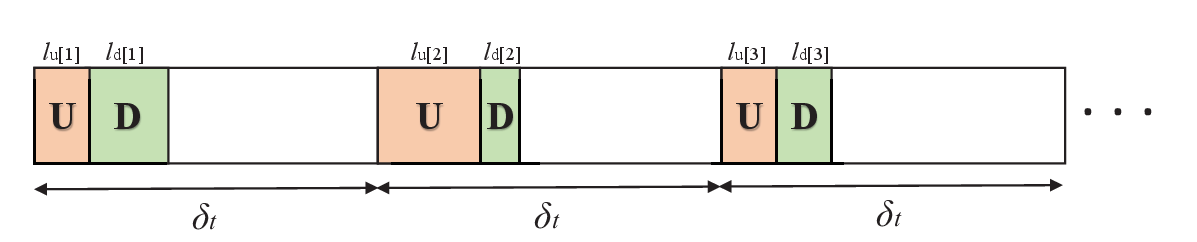}
\caption{\textcolor{black}{Illustration of periodic short-packet UAV relaying, where $\mathbf{U}$ and $\mathbf{D}$ indicate uplink and downlink transmissions, respectively. Sensitive information is generated and made available for transmission at the beginning of each timeslot $\delta_t$, with the remaining duration devoted to sensing and collecting data from the environment.
}}
\label{fig2}
\end{figure}
\textcolor{black}{We assume that the UR-aided DF relaying for SPC occurs at the start of each timeslot $\delta_t$, as illustrated in Fig. \ref{fig2}. We assume that confidential information is generated and made available for secure SPC, while the rest of the timeslot is dedicated to sensing the environment and gathering information by Alice for the next transmission via UAV-assisted DF relaying. Recall that DF relaying consists of two phases. In the first phase (i.e., uplink transmission), Alice transmits one short packet containing sensitive information to the UR over $l_u[n]$ channel uses, where $n=\{1, 2, \cdots\}$ denotes the index of each timeslot, and then the UR decodes the received signal to obtain the transmitted confidential message. In the second phase (i.e., downlink transmission), the UR encodes the message with a different codebook for security purposes, forwarding it to the desired destination Bob over $l_d[n]$  channel uses and Bob retrieves the original confidential information. In the meantime,  Eve wiretaps the ongoing transmissions to obtain confidential data and thus poses a security threat.
Note that end-to-end SPC generally occupies much less time than one timeslot, i.e. $\delta_i[n] \ll \delta_t~\forall n$ and $i\in\{u,~d\}$, where $\delta_u[n]$ and $\delta_d[n]$ indicate the duration required by a finite blocklength SPC for the uplink and downlink, respectively. Since the time taken for one channel use is inversely proportional to the available bandwidth, we have $\delta_i[n] = \frac{l_i[n]}{W},~\forall n$.}

\subsection{System Assumptions and Constraints}
Without loss of generality, we consider a 3D Cartesian coordinate system, where Alice, Bob, and Eve are respectively located on the ground at $\q_a=[x_a,y_a,0]^T$, $\q_b=[x_b,y_b,0]^T$, and $\q_e=[x_e,y_e,0]^T$, where $[\cdot]^T$ represents the transpose. While the location of the legitimate network nodes is perfectly known, Eve attempts to hide her location via passive eavesdropping. Nonetheless, we assume that an approximate estimate of Eve's location\footnote{Some information may be available about Eve's location due to geographical constraints, prior information about adversarial operations, or unintended emissions from Eve's RF electronics.}
within a given uncertainty region can be obtained and tracked such that
\begin{align}\label{uncertainty}
    \| \q_e - \tilde{\q}_e\| \leq \Delta_e,
\end{align}
where $ \tilde{\q}_e$ is  an estimate of Eve's 3D location within a sphere of radius $\Delta_e$, and $\|\cdot\|$ represents the Euclidean norm. 

\subsubsection{UAV trajectory constraints}
We assume that the UR's flight time horizon is set to $T$ seconds, and is divided into
$N$ timeslots such that  $T= N \delta_t$, and the timeslots are indexed by $n\in\N=\{1, 2, \cdots, N\}$. Since the SPC duration is small,  we can assume that the UR's location over the transmission phase in each timeslot remains approximately unchanged, but varies from one timeslot to another. 
Therefore, the UR's 3D location in timeslot $n$  can be denoted by $\q_r[n]=[x[n], y[n], z[n]]^T~\forall n\in\N$. \textcolor{black}{With this setting, the UR's continuous trajectory can be approximated by a 3D sequence consisting of $N$-waypoints  $\{\q_r[n]\}^{N}_{n=1}$.} UAVs can generally control their horizontal and vertical speeds independently \cite{You2019}. Thus, assuming that the UR's initial and final locations are denoted by $\q_i=[x_i, y_i, z_i]^T$ and $\q_f=[x_f, y_f, z_f]^T$, respectively, 
the following mobility constraints are  imposed on the UR's 3D trajectory:
\begin{subequations}
\begin{align}\label{const:mob}
&\cst{1}:~\q_r[1] = \q_i,\quad \q_r[N] = \q_f,\\
\begin{split}
&\cst{2}:~\sqrt{(x[n+1]-x[n])^2+(y[n+1]-y[n])^2} \\
&\qquad\quad\leq v^{max}_{xy} \delta_t,~ \forall n\in \N\setminus N
\end{split}\\
&\cst{3}:~\|z[n+1] - z[n] \| \leq v^{max}_{z} \delta_t,~ \forall n\in \N\setminus N\\
&\cst{4}:~ H^{min} \leq z[n]  \leq H^{max},~ \forall n\in \N
\end{align}
\end{subequations}
where \cst{2} and \cst{3} limit horizontal and vertical displacements of the UR for consecutive timeslots, and $v^{max}_{xy}$ and $v^{max}_{z}$ denote the maximum velocity of the UR in the horizontal and vertical directions, respectively. Furthermore, \cst{4} indicates that the altitude of the UR needs to be larger than a minimum required height $H^{min}$, to avoid collision with buildings or other obstacles \cite{Zeng2019}, and smaller than a maximum permitted height $H^{max}$.


\subsubsection{Channel modeling}
The AG channels are assumed to be dominated by path-loss with negligible fading \cite{TatarMamaghani2020}. Thus, for the Alice-UR link, the UR-Eve link, and the UR-Bob link, denoted as $h_{ra}[n]$, $h_{re}[n]$, and $h_{rb}[n]~\forall n$, respectively,  we express their LoS-dominant channel power gains as
\begin{align}
    h_{rj}[n]=\frac{\beta_0}{\|\q_r[n] - \q_j\|^2},~ \forall n\in \N, ~ j \in\{a, e, b\}
\end{align}
where $\beta_0$ denotes the path-loss at a reference distance $d_0$ under omnidirectional propagation, and is given by $\beta_0=\left(\frac{C}{4\pi d_0 f_c}\right)^2$, with $C$ being the speed of light and $f_c$ denoting  the carrier frequency. Furthermore, since both Alice and Eve are terrestrial nodes, the channel model for the Alice-Eve link constitutes both distance-dependent attenuation and small-scale fading \cite{Zhang2021}, the power gain of which can be represented as
\begin{align}
    h_{ae}[n]= \frac{\beta_0}{\|\q_a - \q_e\|^\alpha}\zeta[n],~\forall n\in \N
\end{align}
where $\zeta[n]$ is a unit-mean exponential random variable, i.e., $\zeta[n] \sim \mathbf{Exp}(1)$, that accounts for independent and identically distributed (i.i.d) Rayleigh fading, and $\alpha$ is the corresponding environmental path-loss exponent, with a typical range between $2 < \alpha \leq 4$.

{\it{Remark:} In this work, we assume that $h_{ra}[n]$ and $h_{rb}[n]~\forall n$ are perfectly known by Alice and the UR, respectively, through channel reciprocity and  training. However, it is assumed that only statistical information about Eve's channel is available at the legitimate transmitters. We further assume that Eve has perfect knowledge of the CSI from Alice and the UR to herself, i.e.,  $h_{ae}[n]$ and $h_{re}[n]~\forall n$, respectively, which can be considered as a worst-case scenario from the point of secrecy.}

\subsubsection{Constraints on radio resources and latency}
Since the channel conditions and the UR's location are assumed to remain stable during the short-packet transmission in each timeslot but can change between timeslots, 
we allow the transmit power to be chosen according to the channel condition of a timeslot. In other words, the transmit power of Alice and the UR can change from one timeslot to another. Furthermore, we impose a total power constraint over all timeslots for both Alice and the UR. Accordingly, denoting $\{p_a[n], \forall n\}$ and $\{p_r[n], \forall n\}$ as the transmit power per channel use in timeslot $n$ for Alice and the UR, respectively, the total power budget constraints can be expressed as 
\begin{subequations}
\begin{align}
&\cst{5}:~\sum_{n=1}^{N}p_a[n]l_u[n] \leq  {P}^{tot}_a,\label{const:avepow_alice}\\
&\cst{6}:~\sum_{n=1}^{N}p_r[n]l_d[n] \leq  {P}^{tot}_r, \label{const:avepow_uav}
\end{align}
\end{subequations}
where ${P}^{tot}_a$ and ${P}^{tot}_r$ represent the maximum available power budget for Alice and the UR over the mission time, respectively. The constraints \eqref{const:avepow_alice} and \eqref{const:avepow_uav}  ensure that the amount of energy utilized to transmit short packets of confidential information by Alice and the UR will be sufficient for the duration of the mission. In addition, it is  common practice to design SPC systems with appropriate power control mechanisms to limit the  transmit power in each timeslot, i.e.,  within the peak power limit, which guarantees the reliable and safe operation of the communication system \cite{Cui2018}. Accordingly,  we further adopt the following peak power constraints:
\begin{subequations}
\begin{align}
&\cst{7}:~0 \leq p_a[n] \leq P^{max}_a,~ \forall n\in \N\\
&\cst{8}:~0 \leq p_r[n] \leq {P}^{max}_r,~ \forall n\in \N
\end{align}
\end{subequations}
where ${P}^{max}_a$ and ${P}^{max}_r$ represent the maximum peak power for Alice and the UR, respectively. 

In SPC, delay tolerance is often a critical factor since, to be useful,  the communication system must deliver sensitive information within a short deadline. On the other hand, a larger blocklength results in a longer transmission duration and increased delay. So, for the considered short-packet delay-sensitive system, the delay requirement can be imposed by constraining the number of total blocklengths per transmission to be less than a maximum allowable end-to-end delay, which can be expressed mathematically as
\begin{align}
    &\cst{9}:~\sum_{i} l_i[n] \leq L^{max},~ i\in\{u, d\},~ \forall n\in \N\\
    &\cst{10}:~l_i[n] \in \mathbb{Z}^+,~ i\in\{u, d\},~ \forall n\in \N
\end{align}
where $L^{max}$ denotes the maximum end-to-end latency. 

\subsection{Secrecy Metric for SPC}
Human-centered secure communication systems usually assume channel coding with a sufficiently large (infinite) blocklength, and then consider the so-called \textit{secrecy capacity}  as the performance metric. With positive secrecy capacity, the legitimate source-destination pair can achieve perfectly reliable and secure communications based on the well-known wiretap coding scheme \cite{Mukherjee2014}. On the other hand,  for machine-type SPC, perfect secrecy cannot be guaranteed  due to the assumption of  finite blocklength transmissions. In light of this, \cite{Yang2019} addressed the achievable SR for SPC given the legitimate receiver's decoding error probability and tolerable information leakage to the illegitimate receiver, which we will exploit in the following analysis.

\subsubsection{SR for uplink SPC}
In the uplink transmission of timeslot $n$, Alice generates a short packet and transmits it over $l_u[n]$ channel uses to the UR, while Eve attempts to passively overhear the transmission. According to \cite{Yang2019}, given a reliability constraint on the UR's decoding error probability, denoted by $\varepsilon_r$, and  a security constraint in terms of information leakage to Eve, denoted by $\eta_e$, we express the achievable average SR in bits per channel use for the finite blocklength uplink transmissions in timeslot $n$, denoted by $R^{sec}_u[n]$, as
\begin{align}\label{secrate_1}
    \begin{split}
        R^{sec}_u[n] &= \E_{\zeta[n]}\bigg\{\Big[C^{sec}_u[n]- \sqrt{\frac{V(\gamma_r[n])}{l_u[n]}}\Qi\left( \varepsilon_r\right) \\
        &\qquad- \sqrt{\frac{V(\gamma_{ae}[n])}{l_u[n]}}\Qi\left( \eta_e\right) \Big]_{+}\bigg\},~\forall n\in \N
    \end{split}
\end{align}
where $C^{sec}_u[n]=\log_2\left(\frac{1+\gamma_r[n]}{1+\gamma_{ae}[n]}\right)$ indicates the uplink secrecy capacity with infinite blocklength in timeslot $n$, 
$[x]_+=\max\{x,0\}$, $\E_x\{\cdot\}$ indicates expectation over the random variable $x$, and $\Qi(x)$ is the inverse of the complementary Gaussian cumulative distribution function $Q(x)$, defined as $Q(x)=\int^{\infty}_{x}\frac{1}{\sqrt{2\pi}}\e^{-\frac{r^2}{2}}dr$. Further, $\gamma_r[n]$ and $\gamma_{ae}[n]$, denoting the received signal-to-noise ratios (SNRs) at the UR and Eve in timeslot $n$, are given respectively by
\begin{align}
\gamma_r[n] &= \frac{p_a[n]h_{ra}[n]}{\sigma^2_r[n]} = \frac{p_a[n]\rho_{r}[n]}{\|\q_r[n] - \q_a\|^2},~\forall n\in \N\\
\gamma_{ae}[n]&= \frac{p_a[n]h_{ae}[n]}{\sigma^2_e[n]} = \frac{p_a[n]\rho_{e}[n]}{\|\q_a - \q_e\|^\alpha}{\zeta}[n],~\forall n\in \N
\end{align}
where $\sigma^2_r[n]$ and $\sigma^2_e[n]$ denote the additive white Gaussian noise (AWGN) power at the UR and Eve in timeslot $n$, respectively, and $\rho_{r}[n] = \frac{\beta_0}{\sigma^2_r[n]}$ and $\rho_{e}[n] = \frac{\beta_0}{\sigma^2_e[n]}$. Furthermore, the function $V(\cdot)$ indicates the statistical variation of the channel (a.k.a the channel dispersion), which can be mathematically expressed, according to  \cite{Yang2019},  as
\begin{align}
   V(\gamma) = \log^2_2(\e)\left[1-\left(1+\gamma\right)^{-2}\right],
\end{align}
As such, $V(\gamma_r[n])$ and $V(\gamma_{ae}[n])$ represent the corresponding channel dispersion in timeslot $n$ for the Alice-UR and Alice-Eve's links, respectively. 
Note that the channel dispersion is a monotonically increasing  function of SNR.

\subsubsection{SR for downlink SPC}
In the downlink transmission of timeslot $n$, the UR decodes the confidential information bits transmitted by Alice, encodes them with a different codebook, and forwards the generated short packet towards Bob over $l_d[n]$ channel uses.  While Eve can still wiretap the transmissions from the UR, since the signals from both Alice and the UR have been encoded with different codebooks, she cannot take advantage of a diversity combining strategy by performing, for example, maximum ratio combining (MRC) to improve her reception quality. Accordingly, given Bob's decoding error probability $\varepsilon_b$, the achievable  SR for short-packet downlink transmissions in timeslot $n$, denoted by $R^{sec}_d[n]$, is given by
\begin{align}\label{secrate_2}
  \begin{split}
       R^{sec}_d[n] &= \bigg[ C^{sec}_d[n]- \sqrt{\frac{ V(\gamma_b[n])}{l_d[n]}}\Qi\left( \varepsilon_b\right) \\
       &\qquad- \sqrt{\frac{ V(\gamma_{re}[n])}{l_d[n]}}\Qi\left( \eta_e\right) \bigg]_{+},~ \forall n\in\N
  \end{split} 
\end{align}
where $C^{sec}_d[n]=\log_2\left(\frac{1+\gamma_b[n]}{1+\gamma_{re}[n]}\right)$ specifies the downlink secrecy capacity with infinite blocklength in timeslot $n$, and $\gamma_b[n]$ and $\gamma_{re}[n]$ denote  the received SNRs at Bob and Eve in timeslot $n$, given respectively by
\begin{align}
\gamma_b[n] &= \frac{p_r[n]h_{rb}[n]}{\sigma^2_b[n]} = \frac{p_r[n]\rho_{b}[n]}{\|\q_r[n] - \q_b\|^2},~\forall n\in \N\\
\gamma_{re}[n]&= \frac{p_r[n]h_{re}[n]}{\sigma^2_e[n]} = \frac{p_r[n]\rho_{e}[n]}{\|\q_r[n]- \q_e\|^2},~\forall n\in \N
\end{align}
where $\sigma^2_b[n]$ denotes the power of the AWGN at Bob and  $\rho_{b}[n] = \frac{\beta_0}{\sigma^2_b[n]}$.
Additionally, $V(\gamma_{b}[n])$ and $V(\gamma_{re}[n])$ indicate the channel dispersion for the UR-Bob and UR-Eve links in timeslot $n$, respectively.

\section{Problem Formulation}\label{sec:problem}
In this section, we aim to optimize the secrecy performance of our proposed UAV-aided mobile relay with short-packet transmission by designing the transmit power for Alice and the UR $\P=\left\{\p_a=\{p_a[n],~\forall n\},~\p_r=\{p_r[n],~\forall n\}\right\}$, the transmission blocklengths $\LL=\{l_u[n],~l_d[n],~\forall n\}$, and the UR's 3D trajectory $\Q=\{\q[n], ~\forall n\}$. Assuming that Alice and the UR securely encode the short transmit packets in timeslot $n$ to sustain the desired reliability and security requirements $(\varepsilon_r, \varepsilon_b, \eta_e)$, we define the secrecy throughput metric as the rate of the effective number of securely transmitted information bits in bits-per-second (bps) at timeslot $n$ as
\begin{align}
    \bar{B}_s[n] =\frac{1}{\delta_t}{\min\left\{R^{sec}_u[n]l_u[n] (1- \varepsilon_r), R^{sec}_d[n]l_d[n] (1- \varepsilon_b) \right\}},
 \end{align}

The resulting optimization maximizes the Effective Average Secrecy Throughput (\east) over the mission duration and can be  formulated as
\begin{align}\label{opt_prob}
\Prob{}:& \stackrel{}{\underset{\{\P, \Q, \LL\}}{\mathrm{max}}}~~ \east = \frac{1}{N}\sum_{n=1}^{N} \bar{B}_s[n]\nonumber\\
&~~~~~~\text{s.t.}~~~~~ \cst{1}-\cst{10}.
\end{align}
Note that \Prob{} is a nonconvex optimization problem due to the nonconvex objective function with nonsmooth operator $[\cdot]_+$, highly coupled optimization variables, and nonconvex constraints \cst{5}$-$\cst{6} and \cst{9}. Thus, it is too challenging to solve optimally. \textcolor{black}{We note that the nonsmoothness of the objective function in \Prob{} can be handled since at the optimal point, $\bar{B}_s[n] \ge 0$ must hold. Otherwise, setting $p_a[n]=0$ and/or $p_r[n]=0$ in the given timeslot leads to $\bar{B}_s[n]=0$, which violates the optimality of the solution. In light of this, we can remove the nonsmoothness operator from the objective function without impacting the optimal solution.}
Furthermore, some terms in the objective function are implicit due to $\E_{\zeta[n]}\{\cdot\}$ in $R^{sec}_u[n]$ as well as Eve's location uncertainty. In the following lemmas, we tackle the aforementioned problems to make \Prob{} more tractable.

\begin{lemma}\label{lemma_1}
A closed-form lower-bound expression on the uplink short-packet secure transmission \eqref{secrate_1} can be obtained as
\begin{subequations}
\begin{align}
\hspace{-3mm}  R^{sec}_u[n] &\geq 
        A_0[n] +\E_{\zeta[n]}\{-\log_2(1+\gamma_{ae}[n])\}  \nonumber\\
       &\qquad+A_1[n]\E_{\zeta[n]}\left\{-\sqrt{1-(1+\gamma_{ae}[n])^{-2}}\right\},\\
    &\geq \log_2\left(\frac{1+\gamma_r[n]}{1+\bar{\gamma}_{ae}[n]}\right)- \sqrt{\frac{V(\gamma_r[n])}{l_u[n]}}\Qi\left( \varepsilon_r\right)\nonumber\\
    &\qquad- \sqrt{\frac{V(\bar{\gamma}_{ae}[n])}{l_u[n]}}\Qi\left( \eta_e\right),~ \forall n\in\N
\end{align}
\end{subequations}
where $\bar{\gamma}_{ae}[n] =  \frac{p_a[n]\rho_{e}[n]}{\|\q_a - \q_e\|^\alpha}~\forall n$, and $A_0[n]$ and $A_1[n]$ are given respectively by
\[A_0[n] = \log_2(1+\gamma_{r}[n])- \sqrt{\frac{V(\gamma_r[n])}{l_u[n]}}\Qi\left( \varepsilon_r\right),~\forall n\in \N \]
\[A_1[n] = \frac{\log_2 \e}{\sqrt{l_u[n]}}\Qi(\eta_e),~\forall n\in \N 
\]
\begin{proof}
The proof follows from Jensen's inequality, a fundamental theorem in mathematics which states that $\E\{f(x)\} \geq f(\E\{x\})$ for a convex function $f(x)$, and considering the convexity of the functions $f_1(x)=\max\{x,0\}$, $f_2(x)=-\log_2(1+x)$, and $f_3(x)=-\sqrt{1-(1+x)^{-2}}$.
\end{proof}
\end{lemma}

Now, we deal with the uncertainty in Eve's location by applying a worst-case analysis to facilitate the optimization problem.
\begin{lemma}\label{lemma_2}
Lower bounds on the achievable short-packet  SR in the downlink and uplink transmissions can be obtained as
\begin{align}
    R^{sec}_{u}[n] &\geq \log_2\left({1+\gamma_r[n]}\right)  - \sqrt{\frac{V(\gamma_r[n])}{l_u[n]}}\Qi\left( \varepsilon_r\right)
    \nonumber\\
    &\qquad-\log_2\left({1+\tilde{\gamma}_{ae}[n]}\right) - \sqrt{\frac{V(\tilde{\gamma}_{ae}[n])}{l_u[n]}}\Qi\left( \eta_e\right),\nonumber\\
    &\qquad\treq \tilde{R}^{sec}_{u}[n],~ \forall n\in\N \label{uplink_lb}\\
     R^{sec}_d[n]& \geq  \log_2\left({1+\gamma_b[n]}\right)  - \sqrt{\frac{V(\gamma_b[n])}{l_d[n]}}\Qi\left( \varepsilon_b\right) \nonumber\\
    &\qquad-\log_2\left({1+\tilde{\gamma}_{re}[n]}\right) - \sqrt{\frac{V(\tilde{\gamma}_{re}[n])}{l_d[n]}}\Qi\left( \eta_e\right),\nonumber\\
    &\qquad\treq \tilde{R}^{sec}_d[n],~ \forall n\in\N \label{dllink_lb}\
\end{align}
where $\tilde{\gamma}_{ae}[n]$ and $\tilde{\gamma}_{re}[n]$ are given respectively by
\[\tilde{\gamma}_{ae}[n] = \frac{p_a[n]\rho_{e}[n]}{(\|\q_a - \tilde{\q}_e\| - \Delta_e)^\alpha},~\forall n\in \N\]

\[\tilde{\gamma}_{re}[n] = \frac{p_r[n]\rho_{e}[n]}{(\|\q_r[n] - \tilde{\q}_e\| - \Delta_e)^2},~\forall n\in \N \]

\end{lemma}
\begin{proof}
    The proof follows from Lemma \ref{lemma_1} and considering the fact that $R^{sec}_{u}[n]$ and $R^{sec}_{d}[n]$ are monotonically decreasing functions with respect to (w.r.t.) the terms $\| \q_a - \q_e \|$ and $\| \q_r[n] - \q_e \|$, respectively. Then, the lower bounds given in \eqref{uplink_lb} and \eqref{dllink_lb} are obtained by using the following inequality
  \begin{subequations}
      \begin{align} \label{revtriangle1}
        \| \q_i - \q_e \| &\geq \big|\| \q_i - \tilde{\q}_e\| - \| \tilde{\q}_e - \q_e\|\big|  \\
        &\geq \| \q_i - \tilde{\q}_e\| - \Delta_e,  \label{revtriangle}
    \end{align}    
  \end{subequations}
where $\q_i \in \{\{\q_r[n],~\forall n\}, \q_a\}$. \textcolor{black}{Note that \eqref{revtriangle1} is derived by applying the reverse triangle inequality and \eqref{revtriangle} is obtained from the inequality given in \eqref{uncertainty} by making the reasonable assumption that Eve's uncertainty region, $\Delta_e$, is smaller than the distance between the Alice-Eve and UAV-Eve pairs, i.e., $\|\q_{i} - \tilde{\q}_e\| \geq \Delta_e$.}
\end{proof}

Now, by introducing the slack variable vector $\pmb{\tau}=\{\tau[n], \forall n\}$, \Prob{} can be converted to a more tractable form, whose objective function is differentiable and serves as a lower bound to that of the original problem:
\begin{subequations}
\begin{align}\label{opt_pro_alt}
\Prob{1}:& \stackrel{}{\underset{\{\P,~\Q,~\LL,~\pmb{\tau}\}}{\mathrm{max}}}~~~\frac{1}{T}\sum_{n=1}^{N} \tau[n] \nonumber\\
&\quad\text{s.t.}\quad\cst{1}-\cst{10},\\
&\qquad  \tilde{R}^{sec}_u[n] l_u[n](1- \varepsilon_r) \geq \tau[n],~\forall n\in \N \\
&\qquad  \tilde{R}^{sec}_d[n] l_d[[n](1-\varepsilon_b) \geq \tau[n],~\forall n\in \N 
\end{align}
\end{subequations}
In the sequel, we propose a low-complexity iterative solution to solve the above problem  based on the BSCA algorithm, wherein we optimize each block of variables while keeping the others unchanged in an alternating manner. Such an algorithm  generally approaches a sub-optimal solution while ensuring convergence.

\section {Proposed Low-complexity Design}\label{sec:solution}
In this section,  we divide \Prob{1} into three subproblems: i) joint power optimization for Alice and the UR, ii) short-packet blocklength optimization, and iii) joint optimization of the UAV's motion and altitude, which are tackled separately. Thereafter, we propose an efficient overall algorithm, considering the impact of the block optimization on the joint design. Note that in the sequel, we omit the fixed multiplicative term $\frac{1}{T}$ from the objective function, which does not impact  the proposed solution.

\subsection{Joint Power Optimization}
In this subsection, we  jointly optimize the transmit powers of Alice and the UR, i.e., $\P=\{\p_a,~ \p_r\}$, while keeping other optimization variables fixed. Introducing slack variable vectors $\s=\{s_a[n], s_r[n],~\forall n\}$ and $\pmb{\nu}=\{\nu_a[n], \nu_r[n]~\forall n\}$,  \Prob{1} can be equivalently reformulated as
\begin{subequations}
\begin{align}\label{pow_subprob}
\Prob{2}:& \stackrel{}{\underset{\{\p_a, \p_r, \pmb{\tau}, \s, \pmb{\nu}\}}{\mathrm{max}}}~~~\sum_{n=1}^{N} \tau[n] \nonumber\\
&~~~~~~\text{s.t.}~~~~~ \cst{5}-\cst{8},\\
\begin{split}
&\ln\left(\frac{1+k_{1, j} p_j[n]} {1+k_{2, j} p_j[n]}\right)  \geq k_{3, j}s_j[n]+k_{4, j}\nu_j[n] \\
&\qquad+k_{5, j}\tau[n],~j\in\{a, r\},~\forall n\in\N\label{p2.cst2}
\end{split}\\
\begin{split}
& s^2_j[n] \geq  1-(1+k_{1, j}p_j[n])^{-2},\\
&\qquad s_j[n] \in \R^+,~j\in\{a, r\},~\forall n\in\N\label{p2.cst3}
\end{split}\\
\begin{split}
&\nu^2_j[n] \geq 1-(1+k_{2, j}p_j[n])^{-2},\\
&\qquad \nu_j[n] \in \R^+,~j\in\{a, r\},~\forall n\in\N\label{p2.cst4}
\end{split}
\end{align}
\end{subequations}
where
\[
k_{1, a} = \frac{\rho_r[n]}{\|\q_r[n]-\q_a\|^2},\quad k_{2, a} = \frac{\rho_e[n]}{(\|\q_a-\tilde{\q}_e\|-\Delta_e)^\alpha},\]
\[k_{3, a} = \frac{\Qi(\varepsilon_r)}{\sqrt{l_u[n]}},\quad k_{4, a} = \frac{\Qi(\eta_e)}{\sqrt{l_u[n]}},\quad k_{5, a} = \frac{\ln2}{l_u[n](1- \varepsilon_r)},\]
\[k_{1, r} = \frac{\rho_b[n]}{\|\q_r[n]-\q_b\|^2},\quad k_{2, r} = \frac{\rho_e[n]}{(\|\q_r[n]-\tilde{\q}_e\|-\Delta_e)^2},\]
\[k_{3, r} = \frac{\Qi(\varepsilon_b)}{\sqrt{l_d[n]}},\quad k_{4, r} = \frac{\Qi(\eta_e)}{\sqrt{l_d[n]}},\quad k_{5, r} = \frac{\ln2}{l_d[n](1- \varepsilon_b)}.
\]
Note that \eqref{p2.cst3} and \eqref{p2.cst4} should hold with equality at the optimal point; otherwise, their values can be decreased to improve the objective function, which would violate the optimality of the solution. Nonetheless, subproblem \Prob{2} is still nonconvex due to nonconvex constraints \eqref{p2.cst2}, \eqref{p2.cst3}, \eqref{p2.cst4}. Thus, to tackle \Prob{2}, we replace these constraints with corresponding convex approximations by applying the first-order restrictive law of the Taylor expansion method at the given point \cite{Boyd2006}. Accordingly, at a local point $\p^{lo}_j=\{p^{lo}_j[n],~j\in\{a,~r\},~\forall n\}$, we  write \Prob{2} approximately as
\begin{subequations}
\begin{align}\label{pow_subprob_cvx}
\Prob{2.1}:& \stackrel{}{\underset{\{\p_a, \p_r, \pmb{\tau}, \s, \nu\}}{\mathrm{max}}}~~~\sum_{n=1}^{N} \tau[n] \nonumber\\
 \begin{split}
    &\hspace*{-10mm}\text{s.t.}\quad\cst{5}-\cst{8},\\ 
    &\hspace*{-10mm} \qquad s_j[n] \in \R^+,~\nu_j[n] \in \R^+,~\forall n\in\N
\end{split}\\
&\hspace*{-10mm}\ln (1+k_{1, j} p_j[n])  \geq k_{3, j}s_j[n]+k_{4, j}\nu_j[n]+k_{5, j}\tau[n]\nonumber\\
&\hspace*{-10mm}\qquad+k_{6, j}p_j[n]+k_{7,j},~\forall n\in\N\\
& \hspace*{-10mm}\ln s_j[n] + \ln \left(1+ k_{1, j}p_j[n]\right) \hspace{-1mm}\geq\hspace{-1mm} A_1(p^{lo}_j[n];k_{1, j}) (p_j[n]-p^{lo}_j[n]) \nonumber\\
&\hspace*{-10mm}\qquad+ A_0(p^{lo}_j[n];k_{1, j}),~\forall n\in\N\\
&\hspace*{-10mm} \ln \nu_j[n] + \ln \left(1+ k_{2, j}p_j[n]\right) \hspace{-1mm}\geq\hspace{-1mm} A_1(p^{lo}_j[n];k_{2, j}) (p_j[n]-p^{lo}_j[n]) \nonumber\\
&\hspace*{-10mm}\qquad+ A_0(p^{lo}_j[n];k_{2, j}),~\forall n\in\N
\end{align}
\end{subequations}
where 
\[k_{6, j} = \frac{k_{2, j}}{1+k_{2, j}p^{lo}_j[n]},\] 
\[k_{7,j} = \ln(1+k_{2,j}p^{lo}_j[n])-k_{6, j}p^{lo}_j[n].\] 
Also, the functions $A_0(x; k)$ and $A_1(x; k)$ are defined respectively as
\[A_0(x;k) =\frac{1}{2} \ln \left(kx\left[2+kx\right]\right),~x,k>0\]
\[A_1(x; k) = \frac{kx+1}{x \left(kx+ 2\right)},~x,k>0\] 
Since \Prob{2.1} is convex w.r.t. the optimization variables, it can be efficiently solved using standard convex optimization tools.
 
\subsection{Short-packet Blocklength Optimization}
In this subsection, we optimize the blocklength vectors of both the uplink and downlink SPC, i.e., $\LL=\{l_u[n], l_d[n],~\forall n\}$. As such,  the corresponding subproblem can be expressed as
\begin{subequations}
\begin{align}
\Prob{3}:& \stackrel{}{\underset{\LL,~\pmb{\tau} }{\mathrm{max}}}~~~\sum_{n=1}^{N} \tau[n]\nonumber\\
&\quad\text{s.t.}\quad\cst{5}, \cst{6}, \cst{9}, \cst{10}\label{p3.cst1}\\
&\hspace{-5mm}a_{0,i}l_i[n] - a_{1,i}\sqrt{l_i[n]} \geq \tau[n],~ i\in\{u, d\},~ \forall n\in\N
\end{align}
\end{subequations}
where 
\[a_{0,u} = (1-\varepsilon_r)\log_2\left(\frac{{1+\gamma_r[n]}}{{1+\tilde{\gamma}_{ae}[n]}}\right),\]
\[
a_{0,d} = (1-\varepsilon_b)\log_2\left(\frac{{1+\gamma_b[n]}}{{1+\tilde{\gamma}_{re}[n]}}\right),\]
\[a_{1,u}=(1-\varepsilon_r)\left[\sqrt{{V(\gamma_r[n])}}\Qi\left(\varepsilon_r\right)\hspace{-0.5mm} + \hspace{-0.5mm}\sqrt{V(\tilde{\gamma}_{ae}[n])}\Qi\left( \eta_e\right)\right],
\]
\[a_{1,d}=(1-\varepsilon_b)\left[\sqrt{{V(\gamma_b[n])}}\Qi\left(\varepsilon_b\right) \hspace{-0.5mm}+\hspace{-0.5mm} \sqrt{V(\tilde{\gamma}_{re}[n])}\Qi\left( \eta_e\right)\right].
\]
\Prob{3} is a nonlinear integer programming problem due to \cst{10}. One possible approach to tackle this challenge is to relax \Prob{3} into a convex optimization problem by converting the integer-valued vector $\LL$ to a positive continuous vector $\widetilde{\LL}=\{\tilde{l}_u[n], \tilde{l}_d[n],~\forall n\}$ such that
\begin{subequations}
\begin{align}
\Prob{3.1}:& \stackrel{}{\underset{\widetilde{\LL},~\pmb{\tau}}{\mathrm{max}}}~~~\sum_{n=1}^{N} \tau[n]\nonumber\\
&\quad\text{s.t.}\quad\cst{5}, \cst{6}, \cst{9} \label{p31.cst1}\\
&\hspace{-10mm}\tilde{l}_i[n] \geq 1,~ i\in\{u, d\},~ \forall n\in \N \label{p31.cst2}\\
&\hspace{-10mm} a_{0,i}\tilde{l}_i[n] - a_{1,i}\sqrt{\tilde{l}_i[n]} \geq \tau[n],~i\in\{u, d\},~ \forall n\in\N \label{p31.cst3}
\end{align}
\end{subequations}
Notice that \eqref{p31.cst3}  in \Prob{3.1} is a nonconvex constraint, following the law of the second-order derivative \cite{Boyd2006}. Thus, the convex reformulation of \Prob{3.1} at a local point $\widetilde{L}^{lo}=\{\tilde{l}^{lo}_u[n], \tilde{l}^{lo}_d[n],~\forall n\}$ can be given, using a first-order Taylor approximation, as
\begin{subequations}
\begin{align}\label{sp_subprob_cvx}
\Prob{3.2}:& \stackrel{}{\underset{\widetilde{\LL},~\pmb{\tau}}{\mathrm{max}}}~~~\sum_{n=1}^{N} \tau[n]\nonumber\\
&\quad\text{s.t.}\quad\eqref{p31.cst1}, \eqref{p31.cst2}\\
&\hspace{-10mm}a_{0,i}\tilde{l}_i[n] -  \tau[n] \geq  a_{1,i}\sqrt{\tilde{l}^{lo}_i[n]} \nonumber\\
&\hspace{-10mm}+ \frac{a_{1,i}}{2\sqrt{\tilde{l}^{lo}_i[n]}}(\tilde{l}_i[n]-\tilde{l}^{lo}_i[n]),~i\in\{u,~d\},~\forall n\in\N
\end{align}
\end{subequations}
\Prob{3.2} is a convex optimization problem, and thus it can be solved efficiently. In order to obtain the integer solution to \Prob{3}, denoted by $\LL^\star$, one can simply round down the noninteger solution to \Prob{3.3} such that $\LL^\star= \lfloor\widetilde{\LL}^{opt}\rfloor$, where $\lfloor x\rfloor$ indicates the largest integer less than or equal to $x$. We note that such a rounding approach provides a  lower-bound solution to the original integer problem as the objective function is a monotonically increasing function of the blocklength and none of the constraints is violated.


\subsection{3D UR Trajectory Optimization}
This subsection explores the joint optimization of the UR's motion and altitude. In light of this, we recast \Prob{1} to optimize $\Q$, while keeping the other variables fixed:
\begin{subequations}
\begin{align}\label{trj_subprob}
\Prob{4}:& \stackrel{}{\underset{\{\q_r,~\pmb{\tau}\}}{\mathrm{max}}}~~~\sum_{n=1}^{N} \tau[n] \nonumber\\
&~~~~~~\text{s.t.}~~~~~ \cst{1}-\cst{4},\\
&\log_2(1+\gamma_r[n]) - b_0\sqrt{1-(1+\gamma_r[n])^{-2}}  \geq b_2 \tau[n] \nonumber\\
&+ b_1,~ \forall n\in\N \label{p4.cst2}\\
&\log_2\left(\frac{1+\gamma_b[n]}{1+\tilde{\gamma}_{re}[n]}\right) - c_0\sqrt{1-(1+\gamma_b[n])^{-2}} \nonumber\\
&- c_1\sqrt{1-(1+\tilde{\gamma}_{re}[n])^{-2}} \geq c_2 \tau[n],~ \forall n\in\N \label{p4.cst3}
\end{align}
\end{subequations}
where 
\[b_0 = \frac{\Qi(\varepsilon_r)\log_2 \e}{\sqrt{l_u[n]}},~b_2 = \frac{1}{l_u[n](1-\varepsilon_r)},\]
\[b_1 = \log_2\left({1+\tilde{\gamma}_{ae}[n]}\right) - \sqrt{\frac{V(\tilde{\gamma}_{ae}[n])}{l_u[n]}}\Qi\left( \eta_e\right),\]
\[c_0 \hspace{-1mm}=\hspace{-1mm} \frac{\Qi(\varepsilon_b)\log_2 \e}{\sqrt{l_d[n]}},~c_1 \hspace{-1mm}=\hspace{-1mm} \frac{\Qi(\eta_e)\log_2 \e}{\sqrt{l_d[n]}},~c_2 \hspace{-1mm}=\hspace{-1mm} \frac{1}{l_d[n](1\hspace{-1mm}-\hspace{-1mm}\varepsilon_b)}.\]
We stress that \Prob{4} is still a nonconvex optimization problem due to nonconvex constraints \eqref{p4.cst2} and \eqref{p4.cst3}. In the following, we focus on transforming these constraints into convex approximations to make the problem tractable.
\subsubsection{Convex reformulation of \eqref{p4.cst2}} First, we equivalently write \eqref{p4.cst2} in a more tractable way by introducing nonnegative slack variables $\pmb{\lambda}=\{\lambda_1[n], \lambda_2[n],~\forall n\}$ and $\pmb{\beta}=\{\beta_1[n],~\forall n\}$:
\begin{subequations}
\begin{align}
    &\log_2(1+\lambda_1[n]) - b_0\beta_1[n]  \geq b_2 \tau[n] + b_1,~\forall n\in\N \label{ineq1}\\
    & \lambda_2[n] \geq \frac{\|\q_r[n] - \q_a \|^2}{p_a[n]\rho_r[n]},~\forall n\in\N\label{ineq3}\\
    &\lambda_1[n] \lambda_2[n] \leq 1,~\forall n\in\N\label{ineq2}\\
    &\beta^2_1[n] \geq 1-(1+\lambda_1[n])^{-2},~\forall n\in\N \label{ineq4}
\end{align}
\end{subequations}
We note that constraints \eqref{ineq1} and \eqref{ineq3} are convex, while \eqref{ineq2} and \eqref{ineq4}, introduced to ensure the smoothness of \Prob{4}, are nonconvex. We stress that \eqref{ineq3}-\eqref{ineq4} should hold with equality at the optimal point. Before proceeding further, we mention a fruitful lemma below. 
\begin{lemma}\label{lemma_3}
    Let $f(x,y)=\frac{1}{xy}$ with $x, y > 0$. At any given point $(x_0, y_0)$ in the domain of $f$, the following function serves as a global lower bound on $f(x,y)$, i.e., 
    \begin{align}\label{upperbound}
        f_{lb}(x,y;x_0,y_0) =  -\frac{x\,y_{0}+x_{0}\,y-3\,x_{0}\,y_{0}}{{x_{0}}^2\,{y_{0}}^2} \leq f(x,y).
    \end{align}
\begin{proof}
\textcolor{black}{Please see Appendix \ref{Appendix A}.}
\end{proof}
\end{lemma}
Thus, at a given point $(\q^{lo}_r,\pmb{\lambda}^{lo},\pmb{\beta}^{lo})$, the following serve as convex approximations of the constraints \eqref{ineq2} and \eqref{ineq4}.
\begin{subequations}\label{cst_cvx}
\begin{align} 
 & 1\leq  
  f_{lb}(\lambda_1[n],\lambda_2[n];\lambda^{lo}_1[n],\lambda^{lo}_2[n]),~\forall n\in\N\label{ineq2_cvx}\\
    &\ln(\beta_1[n]) + \ln(1+\lambda_1[n]) \geq   A_0(\lambda^{lo}_1[n];1) \nonumber\\
    &\qquad+ A_1(\lambda^{lo}_1[n];1)\left(\lambda_1[n]-\lambda^{lo}_1[n]\right),~\forall n\in\N \label{ineq5_cvx}
\end{align}
\end{subequations}
Note that \eqref{ineq2_cvx} follows from Lemma \ref{lemma_3} and \eqref{ineq5_cvx} follows from the concavity of the logarithm function.

\subsubsection{Convex reformulation of \eqref{p4.cst3}}
Introducing the nonnegative  slack variables $\pmb{\omega}=\{\omega_1[n], \omega_2[n],~\forall n\}$, $\pmb{\psi}=\{\psi_1[n],~\forall n\}$, $\mathbf{u}=\{u_1[n],~\forall n\}$, and $\mathbf{v}=\{v_1[n], v_2[n],~\forall n\}$, we can reformulate \eqref{p4.cst3} into the equivalent convex constraint 
  \begin{align}\label{p4.cst3_cvx}
    &\stackrel{}{\underset{{\mathcal{I}}}{\underbrace{\log_2\left({1+\omega_1[n]}\right) - c_0\psi_1[n]}}} \nonumber\\
    &\quad-  \stackrel{}{\underset{{\mathcal{II}}}{\underbrace{\log_2\left(1+\frac{1}{u_1[n]}\right) - c_1v_1[n]}}} \geq c_2 \tau[n],~\forall n\in\N
\end{align}
with the additional constraints for Part $\mathcal{I}$, given by
\begin{subequations}\label{partI_cst}
\begin{align}
    & \omega_2[n] \geq \frac{\|\q_r[n] - \q_b \|^2}{p_r[n]\rho_b[n]},~\forall n\in\N\label{34a}\\
         &\omega_1[n] \omega_2[n] \leq 1 ,~\forall n\in\N\\
    &\psi^2_1[n] \geq 1-(1+\omega_1[n])^{-2},~\forall n\in\N 
\end{align}
\end{subequations}
and the extra constraints for Part $\mathcal{II}$ 
\begin{subequations}\label{partII_cst}
    \begin{align}
    & u_1[n] \leq \frac{(\|\q_r[n] - \tilde{\q}_e \| - \Delta_e)^2}{p_r[n]\rho_b[n]},~\forall n\in\N\\
     &v^2_1[n] \geq 1-\left(1+v_2[n]\right)^{-2},~\forall n\in\N \\
     &u_1[n]v_2[n] \leq 1,~\forall n\in\N
     \end{align}
\end{subequations}
We note that some constraints arising from the reformulation of \eqref{p4.cst3_cvx} are nonconvex. Thus, using Lemma \ref{lemma_3} and \cite[Lemma 3]{TatarMamaghani2021}, we obtain  convex approximations of the constraints \eqref{partI_cst} and \eqref{partII_cst} at the  given local point $(\q^{lo}_r,\pmb{\omega}^{lo},\pmb{\psi}^{lo}, \mathbf{v}^{lo}, \mathbf{u}^{lo})$ as
\begin{subequations}\label{cvx_cst_rest}
    \begin{align}
 & 1\leq   f_{lb}(\omega_1[n],\omega_2[n];\omega^{lo}_1[n],\omega^{lo}_2[n]),~\forall n\in\N\\
     &\ln(\psi_1[n]) + \ln(1+\omega_1[n]) \geq   A_0(\omega^{lo}_1[n];1) \nonumber\\
     &\qquad+ A_1(\omega^{lo}_1[n];1)\left(\omega_1[n]-\omega^{lo}_1[n]\right)
,~\forall n\in\N    \\
  & p_r[n]\rho_b[n]u_1[n]+2\Delta_e\|\q_r[n] - \tilde{\q}_e\|  \leq -\|{\q^{lo}_r[n]}\|^2 \nonumber\\
  &\qquad+ 2\left({\q^{lo}_r[n]} - \tilde{\q}_e\right)^T\q_r[n] +d_0,~\forall n\in\N\label{cvx_cstabs}\\
     &  \ln(v_1[n]) + \ln(1+v_2[n]) \geq   A_0(v^{lo}_2[n];1) \nonumber\\
     &\qquad+ A_1(v^{lo}_2[n];1)\left(v_2[n]-v^{lo}_2[n]\right)
,~\forall n\in\N \\
& 
  u_1[n] \geq \frac{1}{v_2[n]},~\forall n\in\N
    \end{align}
\end{subequations}
where  $d_0 =  \|\tilde{\q}_e\|^2 + \Delta^2_e$. We now express the convex reformulation of subproblem \Prob{4} as
\begin{align}\label{trj_subprob_cvx}
\Prob{4.1}:& \stackrel{}{\underset{\{\q_r, \pmb{\tau}, \pmb{\lambda}, \pmb{\beta}, \pmb{\omega}, \pmb{\psi}, \mathbf{u},\mathbf{v}\}}{\mathrm{max}}}~~~\sum_{n=1}^{N} \tau[n] \nonumber\\
&\hspace{-5mm}\text{s.t.}~~~\cst{1}-\cst{4}, \eqref{ineq1}, \eqref{ineq3}, \eqref{cst_cvx}, \eqref{p4.cst3_cvx}, \eqref{34a}, \eqref{cvx_cst_rest}
\end{align}
Since \Prob{4.1} is convex, it can be efficiently solved.

\subsection{Overall Algorithm and Complexity Analysis}
\begin{figure}[!t]
\resizebox{\columnwidth}{!}{%
\resizebox{0.5\columnwidth}{!}{%
 \removelatexerror
  \begin{algorithm}[H]{
  \caption{Overall EAST optimization algorithm}\label{myalgo1}
  1:~\textbf{Initialize:}~
  Iteration index $i=0$, choose a feasible local point ($\P^{(i)}, \widetilde{\LL}^{(i)},  \Q^{(i)}$);\\
  2:~\textbf{Repeat:} \\
  2.1:~Calculate $\east^{(i)}$, update the iteration index $i \gets i+1$;\\
  2.2:~Solve the subproblem \Prob{2.1}, then update $\P^{(i)}=\{\p^{(i)}_a[n], \p^{(i)}_r[n],~\forall n\}$;\\
  2.3:~Given $\left(\P^{(i)},\widetilde{\LL}^{(i-1)},  \Q^{(i-1)}\right)$, solve the subproblem \Prob{3.2}, then update $\widetilde{\LL}^{(i)}=\{\tilde{l}^{(i)}_u[n], \tilde{l}^{(i)}_d[n],~\forall n\}$;\\
  2.4:~Given $\left(\P^{(i)}, \widetilde{\LL}^{(i)},  \Q^{(i-1)} \right)$, initialize slack variables $\left(\pmb{\lambda}^{lo}, \pmb{\beta}^{lo}, \pmb{\omega}^{lo}, \pmb{\psi}^{lo}, \mathbf{u}^{lo}, \mathbf{v}^{lo}\right)$, solve the subproblem \Prob{4.1}, then update $\Q^{(i)}=\{\q^{(i)}_r[n]= [x^{(i)}_r[n], y^{(i)}_r[n], z^{(i)}_r[n]]^T,~\forall n\}$;\\
  2.5:~ Calculate $\east^{(i)}$ at the new local point $\left(\P^{(i)},\widetilde{\LL}^{(i)},  \Q^{(i)}\right)$;\\
  3:~\textbf{Until:} $\|\east^{(i)} - \east^{(i-1)}\| \leq \epsilon$;\\
  4:~\textbf{Return:} $\left(\P^{opt}, \LL^{opt},  \Q^{opt}\right) \gets (\P^{(i)}, \lfloor\widetilde{\LL}^{(i)}\rfloor,  \Q^{(i)})$;}
  \end{algorithm}}}
\end{figure}

\begin{table}[t]
\caption{Complexity analysis}
\vskip -1mm
\centering
\resizebox{\columnwidth}{!}{%
\begin{tabular}{l c} 
 \midrule  \midrule
 Problem & Complexity \\ [0.5ex] 
\midrule 
 Joint power allocation subproblem \Prob{2.1} & $\OO\left((7N)^2(12N+2)^\frac{3}{2}\right)$\\
Short-packet blocklength optimization subproblem \Prob{3.2} & $\OO\left((3N)^2(5N+2)^\frac{3}{2} \right)$\\
 3D UR trajectory design subproblem \Prob{4.1} & $\OO\left((13N)^2(24N+4)^{\frac{3}{2}} \right)$\\
 Overall EAST optimization algorithm  & $\OO\left(N^{3.5}\log_2(\frac{1}{\epsilon})\right)$
\\[1ex] 
\midrule
\end{tabular}}
\label{table:complexity}
\vskip -3mm
\end{table}
\textcolor{black}{In this subsection, we propose an overall iterative algorithm based on the BSCA approach summarized in Algorithm \ref{myalgo1}, which sequentially refines the uplink and downlink transmit power allocation and blocklengths as well as the 3D UAV trajectory way-points to maximize the EAST performance. Specifically, based on Algorithm \ref{myalgo1}, all  optimization blocks are first initialized with feasible values ($\P^{(i)}, \widetilde{\LL}^{(i)},  \Q^{(i)}$) to facilitate the optimization framework. 
Then, optimization blocks are alternatingly updated by solving the corresponding convex subproblems. The convergence criterion ensures that the optimization process terminates when the improvement in EAST becomes marginal, as determined by the convergence accuracy parameter $\epsilon$. It can be proved, similar to \cite{wang2022robust} and \cite{TatarMamaghani2020}, that the proposed alternating algorithm is guaranteed to converge to a local optimum commencing from a feasible point, since the objective function is non-decreasing with each iteration and the~\east~performance metric is finite-valued.}

\textcolor{black}{Moreover, the overall time complexity of Algorithm \ref{myalgo1} depends on the complexity of each convex subproblem and the iteration complexity. The convex conic optimizations are typically solved using the interior-point method, whose complexity can be approximated based on the number of optimization variables and convex constraints. Accordingly, we provide Table \ref{table:complexity} to summarize the computational cost of each subproblem and the complexity of the overall algorithm using big-O notation. The complexity order of Algorithm \ref{myalgo1}, computed approximately as $\OO\left(N^{3.5}\log_2(\frac{1}{\epsilon})\right)$, is polynomial, and thus our proposed design approach can be reasonably implemented for energy-constrained UAV-aided SPC-IoT scenarios.}

\section{Numerical Results and Discussion}\label{sec:numerical}
In this section, we demonstrate  the EAST enhancement achieved by our proposed optimization algorithm  for the considered UAV-aided SPC-IoT scenario. To
exhibit  the effectiveness of our joint trajectory and resource allocation design in Algorithm \ref{myalgo1}, labeled~\prop, we compare it with the following benchmark schemes:
\begin{itemize} 
     \item Benchmark 1: Resource Design with Fixed Trajectory (\rd), where the UR's trajectory is fixed, and the joint transmit power and blocklength optimizations in \Prob{2.1} and \Prob{3.2} are solved sequentially until convergence, resulting in $(\P^{opt}, \LL^{opt})$ to improve EAST.
         
     \item Benchmark 2:      Trajectory Design with Fixed Resources (\td), where the communication resources $(\P, \LL)$ are fixed and only the UR's 3D trajectory is optimized using subproblem \Prob{4.1} in a sequential manner to achieve $\Q^{opt}$, maximizing the EAST performance.
\end{itemize}
Unless otherwise stated, the system parameters are set as in Table \ref{table:sysParams}. Further, the initial feasible trajectory of the UR, $\Q^{(0)}$, is taken to be on a direct line with fixed speed from the initial location to the final location. The uplink and downlink transmission blocklengths per timeslot are initialized as $l^{(0)}_u[n]=l^{(0)}_d[n]=\frac{L^{max}}{2}~\forall n$, and the transmit power of Alice and the UR per timeslot are adapted such that the maximum and total power budget constraints are satisfied, i.e., $p^{(0)}_a[n]=\min\{p^{max}_a, \frac{p^{tot}_a}{Nl^{(0)}_u[n]}\}~\forall n$ and $p^{(0)}_r[n]=\min\{p^{max}_r, \frac{p^{tot}_r}{Nl^{(0)}_d[n]}\}~\forall n$. The slack variables $\left(\pmb{\lambda}^{lo}, \pmb{\beta}^{lo}, \pmb{\omega}^{lo}, \pmb{\psi}^{lo}, \mathbf{u}^{lo}, \mathbf{v}^{lo}\right)$  are initialized such that the corresponding constraints \eqref{ineq2}$-$\eqref{ineq4}, \eqref{partI_cst}, and \eqref{partII_cst}  are met with equality. 

\begin{table}[t]
\caption{Simulation parameters}
\centering
\resizebox{\columnwidth}{!}{%
\begin{tabular}{l c} 
  \midrule \midrule 
 \textbf{Simulation parameter (notation)} & \textbf{Value} \\ 
\midrule 
Alice's total power budget ($p^{tot}_a$) & $1$ kW \\
UR's total power budget ($p^{tot}_r$) & $1$ kW  \\
Alice's maximum transmit power ($p^{max}_a$)& $20$ dBm per channel use\\
UAV's maximum transmit power ($p^{max}_r$)& $20$ dBm per channel use\\
Reference channel power gain ($\beta_0$) & $-70$ dB\\
Terrestrial path-loss exponent ($\alpha$) & $3$\\
Channel noise power ($\sigma^2_r, \sigma^2_b, \sigma^2_e$) & $-140$ dBm\\
Minimum operational altitude ($H^{min}$) & $60$ m\\
Maximum operational altitude ($H^{max}$) & $120$ m\\
Maximum UR's horizontal speed ($v_{xy}$) & $30$ m/s\\
Maximum UR's vertical speed ($v_{z}$) & $5$ m/s\\
UAV's initial location ($\q_i$) & $[-500, -1000, 60]^T$ m\\
UAV's final location ($\q_f$) & $[1000, 500, 60]^T$ m\\
Alice's location ($\q_a$) & $[-700,0, 0]^T$ m\\
Bob's location ($\q_b$) & $[700,0, 0]^T$ m\\
Eve's estimated location ($\tilde{\q}_e$) & $[-500,900, 0]^T$ m\\
Eve's location uncertainty factor ($\Delta_e$) & $10$ m\\
Transmission period ($\delta_t$) & $1$ s\\
Maximum end-to-end latency tolerance ($L^{max}$) & 400\\
Mission time ($T$) & $100$ s\\
Bob’s decoding error probability ($\varepsilon_b$)& $10^{-3}$\\
UAV’s decoding error probability ($\varepsilon_r$)& $10^{-3}$\\
Eve's security constraint ($\eta_e$) & $10^{-2}$\\
Convergence threshold factor  ($\epsilon$) & $10^{-2}$
\\
\midrule 
\end{tabular}}
\label{table:sysParams}
\end{table}

\begin{figure}[t]
\centerline{\includegraphics[width= \columnwidth]{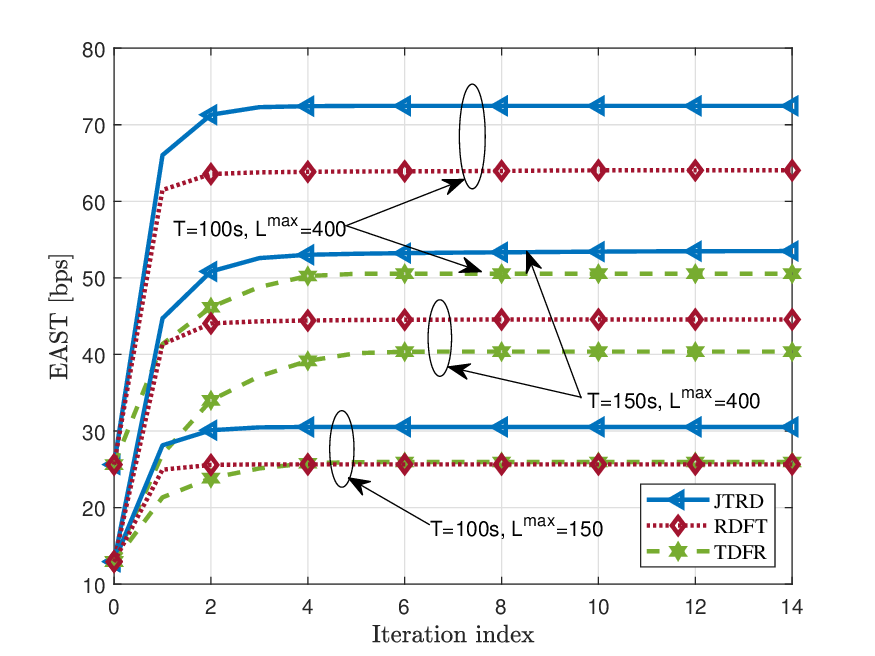}}
\caption{EAST performance vs. iteration index for different designs.}
\label{sim:fig1}
\end{figure}
In Fig. \ref{sim:fig1}, we plot the EAST against the iteration index for all schemes with different mission times $T=\{100, 150\}$ s  and $L^{max}=\{150, 400\}$ to verify the quick convergence of Algorithm \ref{myalgo1} and  the validity of our analysis, as well as to demonstrate the performance advantage of our joint design. We see that all algorithms converge quickly in just a few iterations. Our proposed \prop~approach achieves the best EAST performance in all cases.
For example, \prop~can reach up to $73$ bps, approximately $15\%$ more than \rd, $43\%$ better than \td, and nearly three times the EAST of the initial feasible setting for the case with $T=100$ s and $L^{max}=400$. Fig. \ref{sim:fig1} shows that using the baseline trajectory and optimizing the resource allocation is more important for SPC than optimizing the trajectory for a fixed resource allocation, while the joint design of both is clearly preferable. Furthermore, \rd~outperforms \td~for the considered system setup with $L^{max}=400$, indicating the significance of radio resource management and blocklength optimization when the maximum end-to-end tolerable delay is relatively high. However, reducing the end-to-end delay tolerance, e.g., from  $L^{max}=400$ to $L^{max}=150$, significantly decreases the performance gap between~\td~and~\rd, indicating that the trajectory design becomes critical for a stricter delay requirement. Additionally, we observe from the figure that as the mission time increases from $T=100$~s to $T=150$~s, the EAST performance degrades for all schemes with $L^{max}=400$. This trend can be intuitively explained by the fact that EAST is inversely proportional to the mission time $T$, according to \eqref{opt_prob}. Nevertheless, with a higher $T$, the UR has more flexibility in finding the best locations for secure information relaying, and hence the total amount of information securely exchanged between Alice and Bob is expected to increase, which also contributes to the EAST enhancement. Thus, owing to such a trade-off, it appears that the EAST should asymptotically 
converge to a small but non-zero value as $T$ gets sufficiently large.

\begin{figure}
\centering
\begin{subfigure}{\columnwidth}
  \centering
  \includegraphics[width=\linewidth]{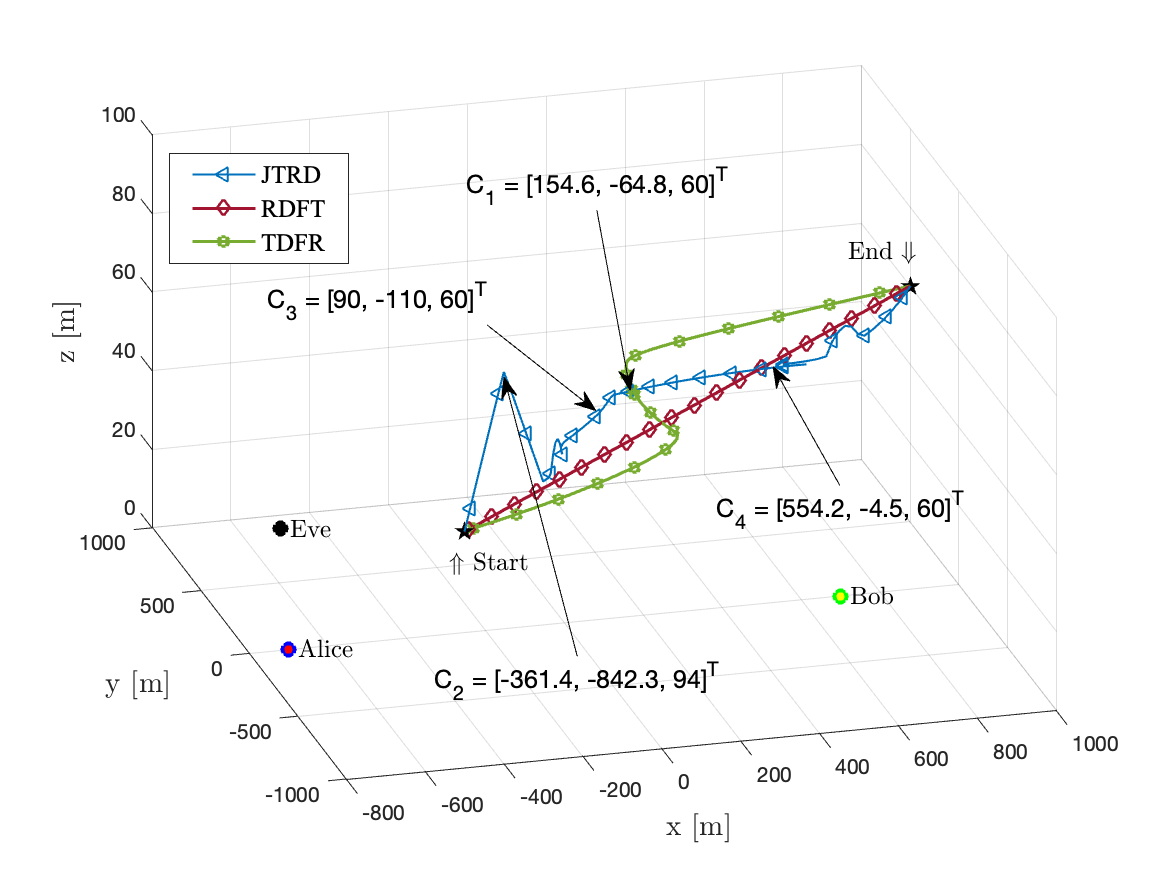}
  \caption{\textcolor{black}{UAV's 3D trajectory profile.}}
  \label{sim:fig2traj}
\end{subfigure}%

\vskip 5pt

\begin{subfigure}{\columnwidth}
  \centering
  \includegraphics[width=\linewidth]{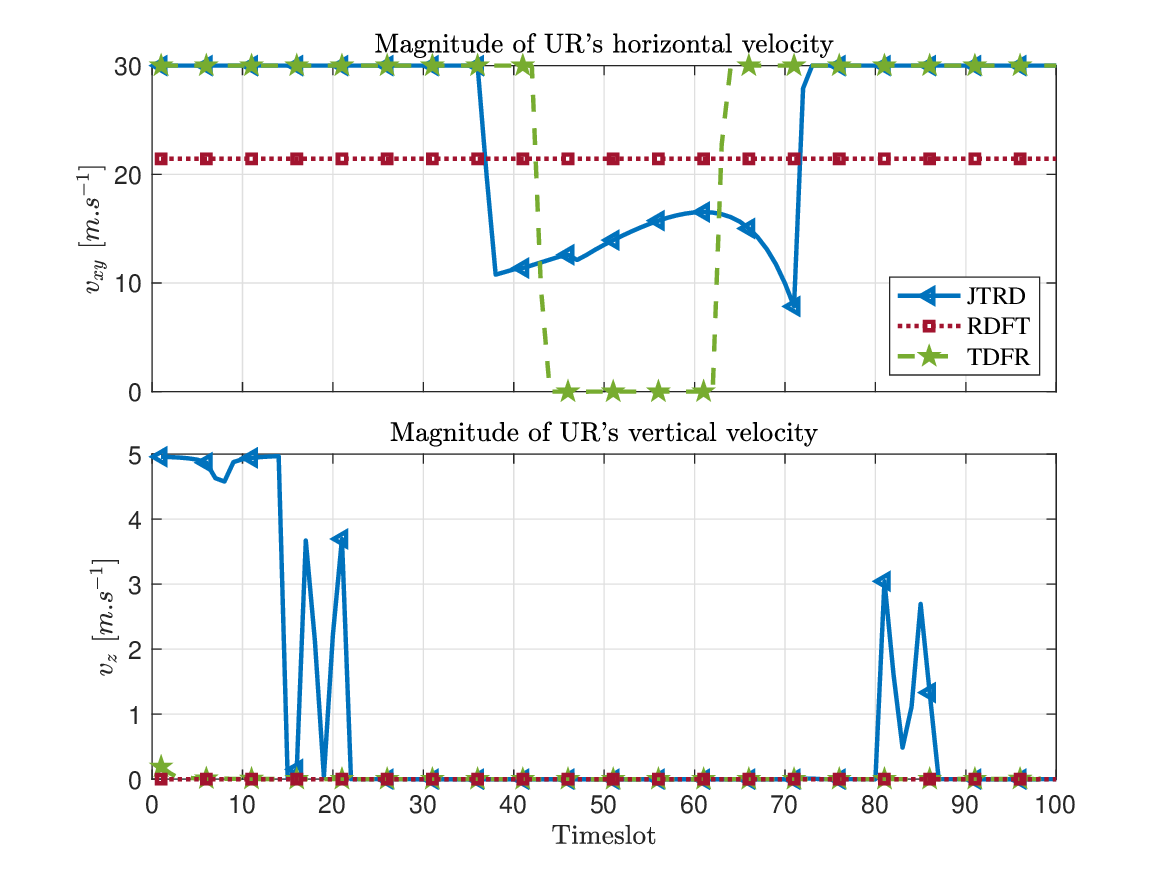}
  \caption{Velocity profile vs. timeslot.}
  \label{sim:fig2vel}
\end{subfigure}
\caption{Designed UAV's 3D trajectory and velocity according to different algorithms.}
\label{sim:fig2}
\end{figure}



\textcolor{black}{Fig. \ref{sim:fig2} illustrates the 3D trajectory and velocity profiles of the UR for different design scenarios based on specific user locations. For \td, we observe from Figs. \ref{sim:fig2}(\subref{sim:fig2traj}) and \ref{sim:fig2}(\subref{sim:fig2vel}) that the UR attempts to fly at the minimum altitude with maximum velocity while heading towards a location between Alice and Bob at the coordinates $\mathbf{C}_1 = [176, -67, 60]^T$ m, and staying aloft at that point as long as possible between $T=44$~s and $T=63$~s. This stretched-S path planning greatly improves the EAST performance compared with the initial direct-path trajectory with fixed velocity, but not as much as \rd~and our proposed~\prop, as demonstrated in Fig. \ref{sim:fig1}.  Nevertheless, when both the trajectory design and resource optimization are taken into account as in the proposed \prop~design, the UR demonstrates a substantially more effective 3D navigation, capitalizing on its altitude adjustment ability to considerably enhance the EAST performance in comparison with the other benchmarks. Specifically, the effective trajectory for the \prop~scheme, according to the path illustrated in Fig. \ref{sim:fig2}(\subref{sim:fig2traj}), requires the UR to fly with full horizontal speed from the starting location $\q_i$ while adaptively adjusting its vertical velocity until the UR reaches a peak at coordinates $C_2=[-361.4, -842.3, 94]^T$ m to further degrade the eavesdropping link. Then the UR sharply decreases its altitude to the constrained minimum value while heading to the position marked as $\mathbf{C}_3=[90, -110,60]^T$ m. A sharp decrease in horizontal velocity occurs at $T=36$~s while maintaining zero vertical velocity, enabling the UR to move with relatively low speed from position $\mathbf{C}_3$ to a location with coordinates $\mathbf{C}_4=[562, -4, 60]^T$ m, to improve the overall EAST performance. Finally, the UR travels along an arc tilted towards Bob while maximally increasing its velocity at time $T=72$~s and maintaining this horizontal velocity so that the last part of the mission from  $\mathbf{C}_4$ to the final location $\q_f$ can be completed by the end of the specified time frame. Overall, this observation highlights the necessity of robust 3D trajectory designs to enhance the secrecy performance of aerial relaying with SPC.}

\begin{figure}
\centering
\begin{subfigure}{\columnwidth}
  \centering
  \includegraphics[width=\linewidth]{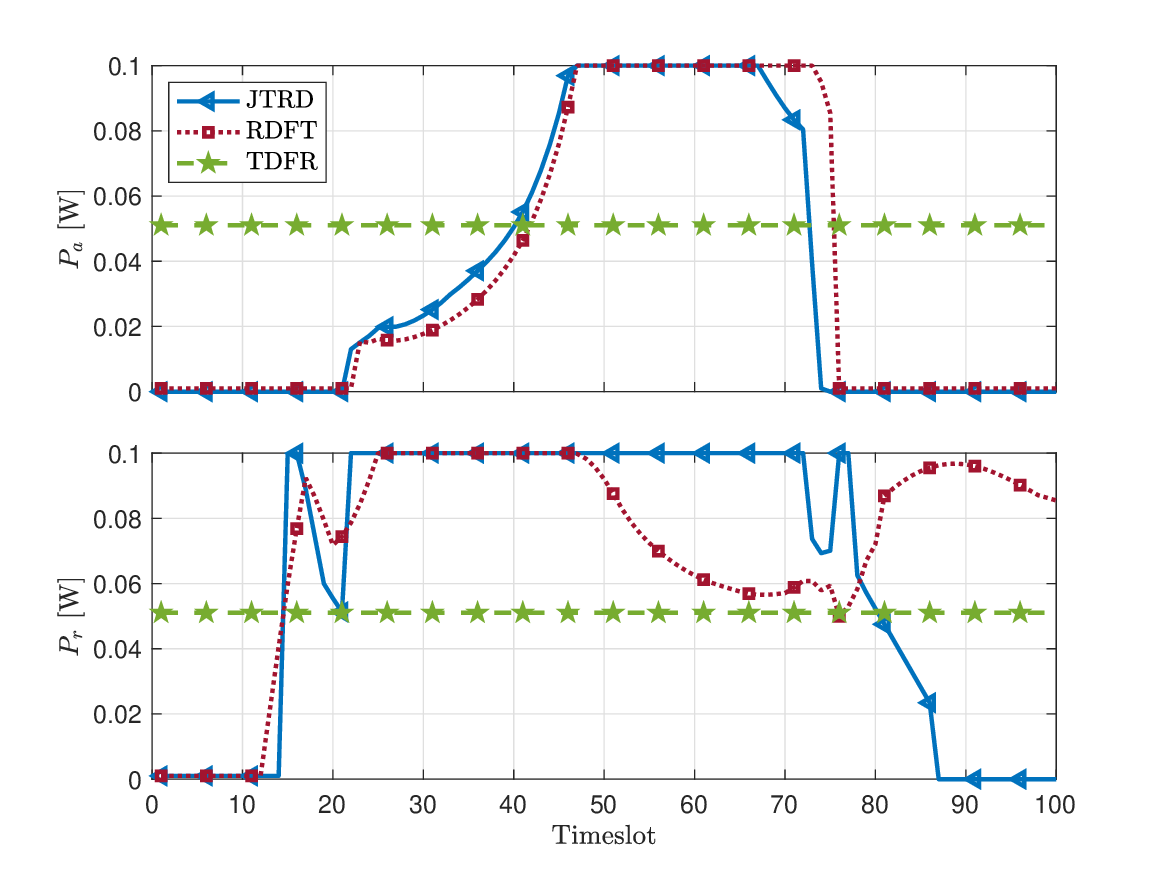}
  \caption{Power allocation vs. timeslot.}
  \label{sim:fig4pow}
\end{subfigure}%

\vskip 5pt

\begin{subfigure}{\columnwidth}
  \centering
  \includegraphics[width=\linewidth]{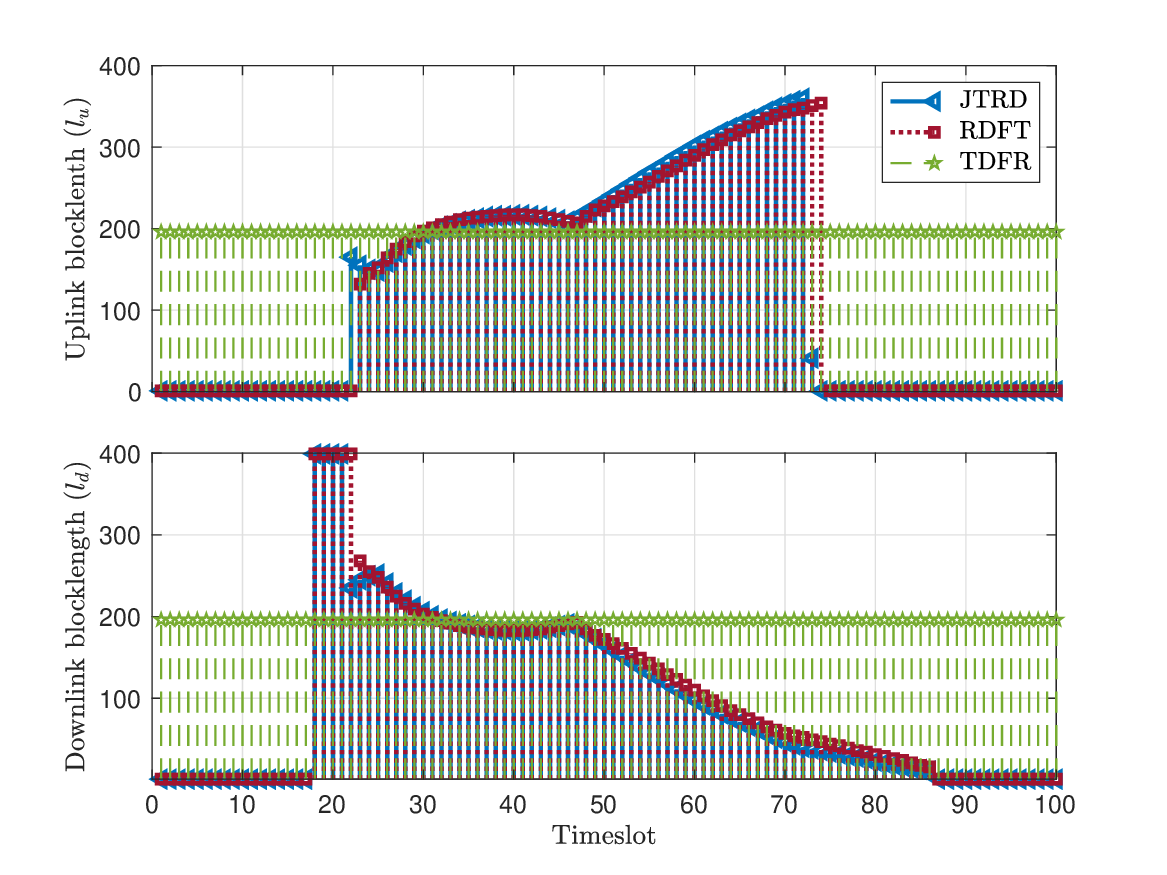}
  \caption{Blocklength vs. timeslot.}
  \label{sim:fig4blocklength}
\end{subfigure}
\caption{Power allocation and optimized blocklength profiles for different designs.}
\label{sim:fig4}
\end{figure}



Fig. \ref{sim:fig4} depicts the optimized transmit power profiles of Alice and the UR,  as well as both the uplink and downlink coding blocklengths for different schemes. It is evident that the adoption of fixed power allocation and equal blocklengths is suboptimal. We can see from Fig. \ref{sim:fig4}(\subref{sim:fig4pow}) that when the channel quality of the main link is worse than the eavesdropping link, proper resource management results in reserving the transmit power and blocklength resources for the moment when better communication channels can be obtained, for example between  timeslots $T=20$ s and $T=74$ s. Furthermore,   we observe from Fig. \ref{sim:fig4}(\subref{sim:fig4blocklength}) that when the UR is farther from Bob,  larger downlink coding blocklengths are adopted, and they reduce in length as the UR approaches Bob. As the UR flies away from Alice, the proposed algorithm efficiently increases Alice's transmit power and uplink blocklength to improve the uplink SR, which ultimately enhances the EAST performance. We note that Alice does not transmit with full power, particularly between timeslots $T=20$ s and $T=45$ s, since the secrecy is jeopardized by the enhancement of the wiretap link. As far as \prop~is concerned, it is generally expected that the UR will maximize the relaying power for the sake of EAST improvement.
\begin{figure}[t]
\centerline{\includegraphics[width= \columnwidth]{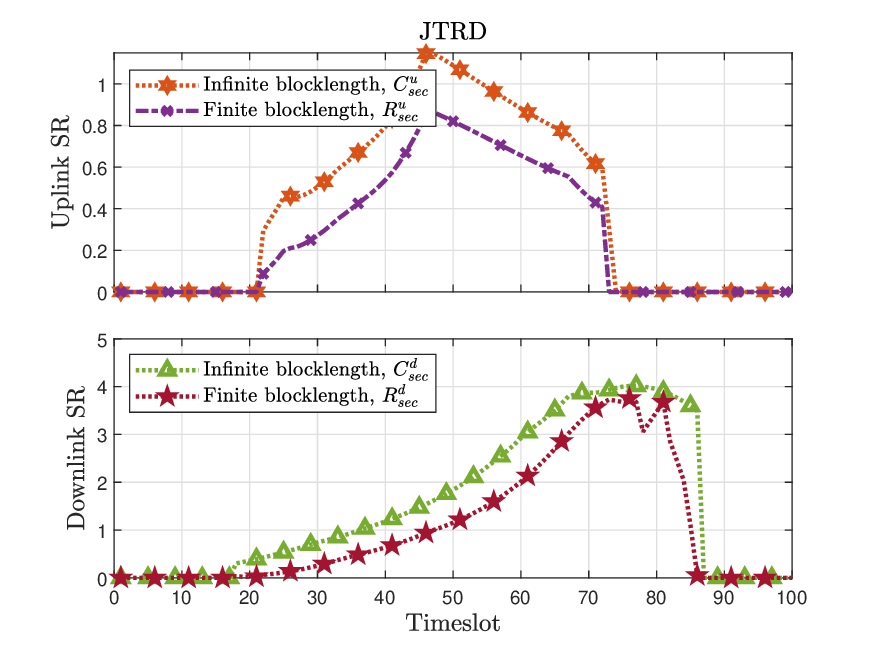}}
\caption{\textcolor{black}{Infinite/finite blocklength SR vs. timeslot for the proposed \prop~scheme.}}
\label{sim:fig6}
\end{figure}

In Fig. \ref{sim:fig6}, the SR of both uplink and downlink transmissions are presented versus time for~\prop~for both the finite and infinite blocklength cases. We can see from the curves that nonzero SR can be achieved when the UAV relay is properly located relative to Alice and Bob for the task of relaying. This goal is achieved between timeslots $T=20$ s and $T=74$ s.
Further, the results reinforce the fact that the adoption of SPC leads to notable a decrease in the achievable SR compared with the conventional infinite blocklength assumption. 
This underscores the need for a different approach to system design for SPC scenarios as compared to conventional systems, and emphasizes the importance of carefully considering system parameters in order to achieve optimal performance, avoiding dramatically suboptimal designs.

\begin{figure}[t]
\centerline{\includegraphics[width= \columnwidth]{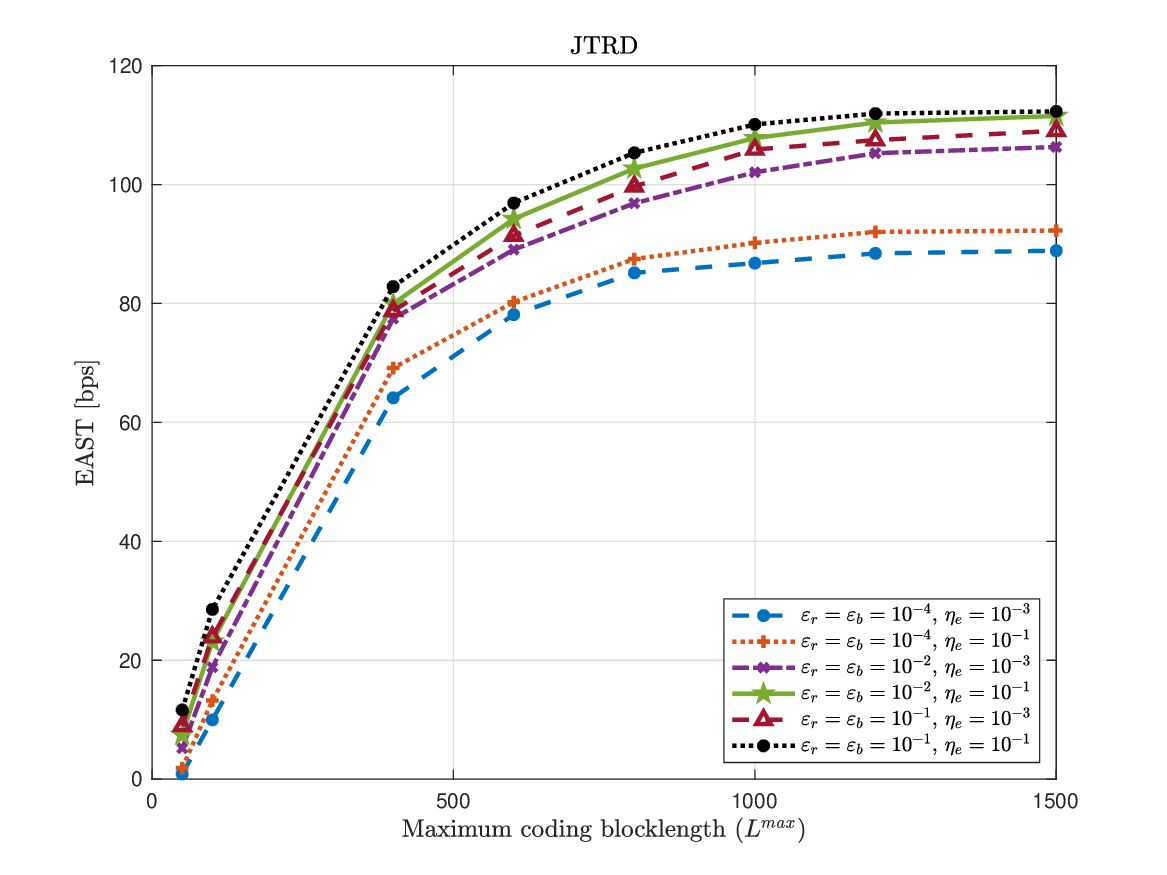}}
\caption{\textcolor{black}{Effect of maximum coding blocklength on the EAST performance of JTRD with different levels of security and reliability.}}
\label{sim:fig7}
\end{figure}
Fig. \ref{sim:fig7} depicts the effects of the end-to-end delay tolerance represented by the maximum coding blocklength $L^{max}$, as well as the reliability requirement and information leakage constraint on the EAST performance in the proposed \prop~scheme. It can be observed from the curves that the larger the maximum blocklength, the higher the EAST up to a certain level for all the settings. Indeed, increasing $L^{max}$ can potentially increase both the uplink and downlink blocklengths, which in turn improves the overall EAST performance by reducing the value of the subtractive terms in the objective function introduced by the finite blocklength. Also, from Fig. \ref{sim:fig7} it can be seen that the EAST experiences a ceiling phenomenon for all the security and reliability levels when $L^{max}$ is sufficiently increased, implying that no further improvement in EAST can be achieved by the joint power and trajectory design due to the power budget limitations. Furthermore, we can observe that when the reliability and/or security requirements are relaxed, as indicated by larger values for $\varepsilon_b$, $\varepsilon_r$, and $\eta_e$, a higher EAST can be achieved by utilizing proper trajectory design and efficient resource management.

\begin{figure}[t]
\centerline{\includegraphics[width= \columnwidth]{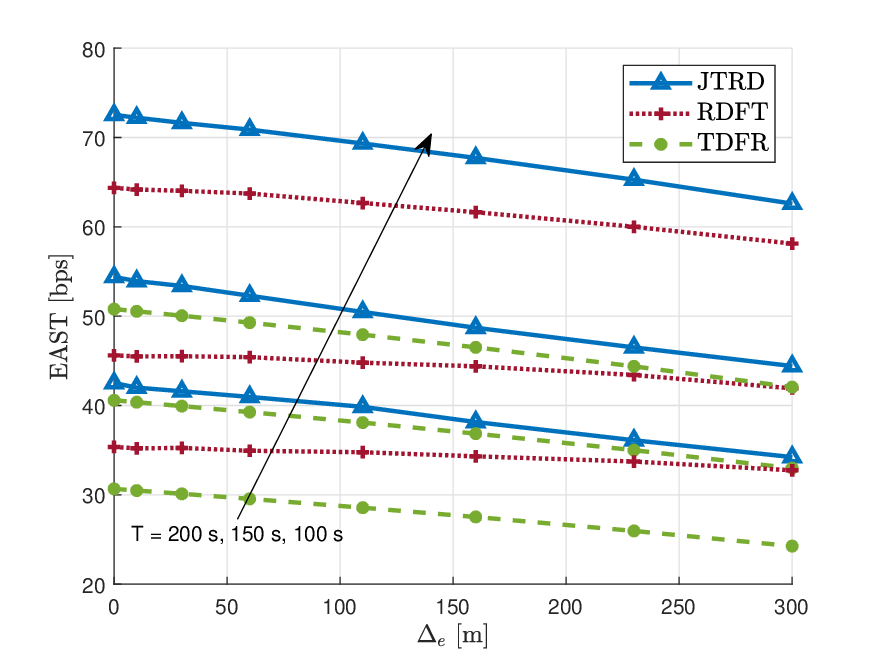}}
\caption{\textcolor{black}{Effect of Eve's location uncertainty on the EAST performance for different schemes.}}
\label{sim:fig8}
\end{figure}
\textcolor{black}{Fig. \ref{sim:fig8} illustrates how the uncertainty in Eve's location impacts the EAST performance for mission times $T=100$ s (relatively low), $T=150$ s, $T=200$ s (relatively high). Note that the longer the mission time, the higher the flexibility in terms of the UR's 3D mobility and resource allocation, which in turn generally improves the effective number of securely exchanged information bits. However, according to its definition, EAST is also inversely proportional to $T$. This inherent trade-off leads to a degradation in EAST as $T$ increases for all considered scenarios, as illustrated in Fig. \ref{sim:fig8}. Furthermore, we see that our proposed \prop~design once again provides the best performance compared with the other methods regardless of the level of Eve's location uncertainty; nonetheless, the performance gap between \prop~and \rd~tends to reduce as $\Delta_e$ becomes larger, which indicates the significance of the uncertainty parameter in the system design. Fig. \ref{sim:fig8} also illustrates that the proposed approach is not highly sensitive to the uncertainty in Eve's location. For example, \prop~experiences only about a 10 bps loss in EAST for $T=100$ s when the location uncertainty increases from the ideal case $\Delta_e=0$ to a relatively high uncertainty, e.g., $\Delta_e=300$ m.  \rd~shows the highest robustness to the variation in $\Delta_e$, but it is less effective in scenarios where $\Delta_e$ is small.}

\section{Conclusions}\label{sec:conclusion}

\textcolor{black}{We proposed the design of a secure UAV-aided SPC system, leveraging a mobile DF relaying mechanism to transmit sensitive short packets from a remote IoT device to its receiver under stringent latency requirements and in the presence of a passive eavesdropper with imperfectly known location. We formulated a nonconvex problem to improve the EAST performance of the system considering security, reliability, latency, and mobility constraints, where we aim to optimize crucial design parameters including uplink and downlink coding blocklengths, transmit powers, and 3D UAV trajectory.
To tackle the inherent nonconvexity of the optimization problem, we introduced a computationally efficient algorithm based on the BSCA approach, effectively dividing the problem into manageable subproblems and solving them separately via convex optimization, which ensures the convergence of the overall algorithm to a locally optimal solution
with relatively low-complexity. Through extensive simulations, we demonstrated the superiority of our proposed \prop~scheme in terms of EAST performance compared to benchmark schemes \rd~and~\td.  Our findings suggest that both uplink and downlink blocklengths should be dynamically adjusted based on the location of the UAV along its trajectory, indicating their pivotal role in improving EAST. In conclusion, our joint design for the proposed four-node system is a promising solution for secure and reliable
communications in UAV-IoT networks with SPC, particularly
in environments with limited power and computing resources.
}

\appendices
\numberwithin{equation}{section}
\makeatletter 
\newcommand{\section@cntformat}{Appendix \thesection:\ }
\makeatother

\section{\textcolor{black}{Proof of Lemma \eqref{lemma_3} }}\label{Appendix A}

\textcolor{black}{To commence the proof, we first calculate the Hessian of $f(x,y)= \frac{1}{xy}$, denoted by $\HH_f$, as
\begin{equation}
\hspace{-3mm}\HH_f \hspace{-0.5mm}=\hspace{-0.5mm} \begin{bmatrix}
\dfrac{\partial^2 f(x,y)}{\partial x^2} & \dfrac{\partial^2 f(x,y)}{\partial x \partial y} \\
\dfrac{\partial^2 f(x,y)}{\partial y \partial x} & \dfrac{\partial^2 f(x,y)}{\partial y^2}
\end{bmatrix}\hspace{-1mm}=\hspace{-1mm} \begin{bmatrix} \frac{2}{x^3\,y} & \frac{1}{x^2\,y^2}\\ \frac{1}{x^2\,y^2} & \frac{2}{x\,y^3} \end{bmatrix}.
\end{equation}
Next, we prove that $\HH_f$ is a positive definite matrix. Note that a Hessian matrix is positive definite if and only if all its eigenvalues are positive. Given that Hessian matrices with continuous second-order derivatives are Hermitian by construction, they possess real eigenvalues. Furthermore, a real symmetric matrix is positive definite if and only if all of its leading principal minors are positive, as per Sylvester's criterion \cite{strang2022introduction}. It is evident that all leading principal minors of $\mathbf{H_f}$, namely the first-order leading principal minor $(\mathbf{H_f})_{1,1}=\frac{2}{x^3y}$ and the second-order leading principal minor (a.k.a. the Hessian determinant) $|\mathbf{H_f}|= \frac{3}{x^4y^4}$, are both positive since the variables $x$ and $y$ are both positive, as per the assumption in Lemma \ref{lemma_3}. With the positive definiteness of $\HH_f$ established, it follows that $f(x, y)$ exhibits joint convexity w.r.t. $x$ and $y$. Thus, since the restrictive first-order Taylor expansion of a convex function provides a global lower bound at the given point $(x_0,y_0)$ \cite{Boyd2006}, one can reach the inequality given by \eqref{upperbound}.}

\bibliographystyle{IEEEtran}

\bibliography{myReferences}

\end{document}